\documentclass[acmsmall,screen,nonacm]{acmart}

\usepackage{amsmath,amssymb,amsthm,mathtools,booktabs}
\usepackage{tabularx}
\usepackage{microtype}
\microtypesetup{expansion=false}

\newtheorem{definition}{Definition}
\newtheorem{theorem}{Theorem}
\newtheorem{proposition}{Proposition}

\newtheorem{example}{Example}
\newtheorem{assumption}{Assumption}
\newtheorem{corollary}{Corollary}
\newtheorem{remark}{Remark}

\newcommand{\AP}{\mathcal{AP}}

\title{Protected Cores Are Not Enough: Certifying AI-Proposed Revisions of Temporal Specifications}
\author{Ruggero Lanotte}
\affiliation{
  \institution{University of Insubria}
  \city{Varese}
  \country{Italy}
}
\setcopyright{none}
\renewcommand\footnotetextcopyrightpermission[1]{}

\begin{document}

\begin{abstract}
Runtime monitoring traditionally evaluates a specification that is fixed before execution or externally modified when requirements change. In learning-enabled and data-intensive systems, however, the temporal relationships represented by a specification may themselves evolve. Allowing an AI component to directly replace a formal specification is unsafe: it may overfit transient behavior, weaken protected requirements, or activate statistically unsupported revisions.

We introduce an intersymbolic architecture in which an untrusted AI proposer suggests temporal specification revisions and a symbolic governor controls their activation. Two results organize the framework. First, origin-version semantics makes the outcome of each obligation invariant to later revisions. Second, aggregate certification can conceal systematic failures on protected triggers; simultaneous aggregate and core-conditional post-selection certification controls both targets. A structural invariant preserves designer-protected components, and a proposer-independent lifetime error bound supports repeated activation decisions. The statistical bound concerns the predictable means of completed certification samples; interpreting it as future operational validity requires an additional stability assumption.
Controlled synthetic experiments use a frozen supervised AI proposer to illustrate the masked-core failure at one decision and across repeated governed revisions. The proposer is a supervised regressor trained offline on synthetic tasks and frozen before use; it ranks candidates by predicted aggregate margin and never observes the protected-trigger success rate, so the masked-core failure arises from optimising the aggregate rather than from an adversary constructed by hand.
\end{abstract}

\maketitle

\section{Introduction}
\label{sec:introduction}

Modern data-intensive systems continuously generate timestamped observations.
Runtime verification provides a principled mechanism for determining whether
such executions satisfy temporal requirements~\cite{Bauer2011Runtime,MalerNickovic2004}.
A typical property may require that an event $A$ be followed by an event $B$
within a bounded interval,
\begin{equation}
\Box\left(A\rightarrow\Diamond_{[a,b]}B\right),
\end{equation}
or, under uncertainty, that such a response occur with probability at least
$\lambda$, as in probabilistic temporal formalisms such as
PrSTL~\cite{SadighKapoor2016},
\begin{equation}
A\Rightarrow_{\lambda,[a,b]}B.
\end{equation}

Such requirements are routinely deployed over streams whose generating
process is not stationary. In a clinical-alarm monitor, for instance, a rule
requiring that a desaturation alert be followed by a documented intervention
within a bounded delay may have been specified for a given patient population,
ward staffing level, and sensor configuration. A change in these conditions
may alter the observed delay or the circumstances in which the rule applies.
Persistent violations then require investigation: the monitor alone cannot
distinguish a mismatch caused by the changed operating regime from a genuine
safety degradation. Similar issues arise in service-level monitoring after a
routing change and in industrial supervision after a retooling. In these
settings, some components of a requirement are safety-relevant and must
survive adaptation, while others describe relationships that may need to
follow the process.

The relationship encoded by a specification may evolve in a non-stationary
environment, a setting closely related to concept-drift adaptation in
streaming learning~\cite{Gama2014ConceptDrift}. For example,
\begin{equation}
A\Rightarrow_{0.90,[2,5]}B
\end{equation}
may cease to describe the observed process, while
\begin{equation}
(A\land C)\Rightarrow_{0.92,[3,8]}B
\end{equation}
becomes supported by persistent evidence.

A fixed monitor may then generate systematic violations after a regime
change. Prior work has therefore considered adaptive monitoring and
changing runtime requirements~\cite{Perotti2015NeuralSymbolic,Carwehl2023ChangingRequirements}.
Conversely, a learner allowed to freely rewrite its own specification may
eliminate violations by weakening the property. Our governor therefore
preserves designer-protected trigger and response components and keeps every
parameter tuple inside a fixed admissible envelope. Structural preservation
alone is insufficient: a high aggregate success rate can conceal systematic
failures on protected triggers. We therefore also require core-conditional
statistical certification.

Proposal, certification, and activation are distinct. At a selection index
$n^\star$, an untrusted proposer supplies candidate revisions using the
observed prefix. The governor checks structural admissibility and collects
completed obligations generated for the selected candidate after selection.
Only if the statistical conditions hold at a later decision index $m$ can the
candidate become active for subsequent events:
\[
\mathsf{Propose}(\mathbf e_{\le n^\star},\psi_j)
\ni\psi^\star
\quad\xrightarrow{\;\text{post-selection certification}\;}
\mathsf{Activate}_{m+1}(\psi^\star),
\qquad m>n^\star.
\]
The proposer may be neural, statistical, rule-based, or LLM-based; its output
is treated as untrusted.

The problem studied in this paper is to decide whether an AI-proposed
revision $\psi$ can replace the active specification $\psi_j$ while ensuring:
(i) deterministic preservation of protected trigger and response structure;
(ii) admissibility of the rewrite; (iii) post-selection statistical support
for the candidate both overall and on protected triggers;
(iv) non-retroactive semantics for obligations generated under earlier
versions; and (v) a global bound on unsupported activations over the
evolution history.

The paper has two main results:
\begin{enumerate}
    \item \emph{Origin-version non-retroactivity.} An obligation retains the
    version under which it was generated, so later revisions cannot change
    its outcome. A retroactive active-version interpretation fails by
    counterexample.

    \item \emph{Masked core failure and its certification remedy.} An
    aggregate success rate may remain high while responses on protected
    triggers fail. Simultaneous aggregate and core-conditional post-selection
    bounds address this failure mode.
\end{enumerate}
The two-time semantics and protected-envelope invariant provide the formal
setting for these results. A proposer-independent lifetime theorem supports
the second result by bounding false certifications relative to the
predictable means of the completed samples available at each decision.
Finite-prefix evidence may initiate evolution but does not create a third
activation guarantee.
Claims about obligations generated after activation require a separate
stability assumption.

The technical novelty is therefore not the Hoeffding inequality or error
spending in isolation. It is the certification target and proof structure
forced by specification evolution: the candidate is chosen adaptively,
evidence is indexed by completed post-selection obligations, activation may
occur while earlier versions still own unresolved obligations, and the
protected-trigger subsequence overlaps the aggregate sample. These features
jointly determine what can be certified without assuming correctness of the
AI proposer.

All proofs are collected in Appendix~\ref{app:proofs}.

\section{Related Work and Positioning}

The proposed framework lies at the intersection of runtime verification, temporal specification learning, concept-drift adaptation, and neural-symbolic or intersymbolic AI~\cite{Platzer2024Intersymbolic}; we use
\emph{intersymbolic} in that sense, with a symbolic and a subsymbolic
component interlinked and the symbolic side retaining decision authority.

\paragraph{Runtime verification and changing requirements.}
Runtime verification monitors execution prefixes against temporal properties~\cite{Bauer2011Runtime,MalerNickovic2004}. Adaptive monitoring is not new: Perotti, Garcez, and Boella combine neural-symbolic monitoring with learning-based adaptation of the initial specification~\cite{Perotti2015NeuralSymbolic}, while Carwehl et al. address runtime verification under changing requirements and preserve intermediate monitoring results when properties are adapted~\cite{Carwehl2023ChangingRequirements}. Carwehl et al.\ address the complementary question of how monitor state is
carried across a property change; the present paper fixes no monitor
representation and is concerned instead with whether a revision may become
active at all. Shared monitoring of multiple temporal properties has also been studied by Demir and Ulus~\cite{LoomRV2026}, while SENTIL supplies a probabilistic temporal-logic monitoring tool~\cite{SENTIL2026}. These two papers address monitoring computation or probabilistic checking; our activation rule additionally certifies a protected-trigger subpopulation under adaptive proposal. Our setting differs in that candidate structural revisions are generated by an explicitly untrusted proposer and become active only after a separate symbolic and statistical certification step.

\paragraph{Learning temporal specifications.}
Temporal-logic learning and specification mining infer formulas or parameters from observed traces. Bombara and Belta introduced online learning of STL formulas from incrementally arriving labelled signals~\cite{BombaraBelta2018} and later developed a decision-tree framework covering both offline and online STL learning~\cite{BombaraBelta2021}. Probabilistic Signal Temporal Logic (PrSTL) represents uncertain temporal properties whose probabilistic predicates can evolve as beliefs are updated from data~\cite{SadighKapoor2016}. We therefore do not claim novelty for online formula learning, probabilistic temporal specifications, or parameter adaptation in isolation. Instead, the learned component in our architecture is only a proposer: activation is governed independently of the proposal mechanism.

\paragraph{Concept drift and sequential certification.}
Concept-drift adaptation studies changes in data generating processes and online model adaptation~\cite{Gama2014ConceptDrift}. Our notion of \emph{semantic drift} is narrower: it concerns a persistent change in the validity of an explicitly represented temporal relation. For repeated statistical decisions, ordinary fixed-time intervals are insufficient under optional stopping; confidence sequences provide time-uniform coverage over unbounded horizons~\cite{HowardEtAl2021}, while betting-based constructions provide practical anytime-valid confidence sequences for bounded means~\cite{WaudbySmithRamdas2024}. We use this sequential-inference perspective to certify a sequence of specification revisions rather than a single model estimate.
Classical post-selection inference instead adjusts inference for a
data-dependent model-selection step~\cite{BerkEtAl2013PostSelection}. Our
construction separates the selected prefix from fresh prospective outcomes
and conditions on the selection-time information; it therefore addresses a
different sequential object, but the same selection-versus-inference
distinction is essential.

\paragraph{Certified self-modification of AI systems.}
The Statistical G\"odel Machine (SGM) admits recursive edits through
statistical tests and a global risk budget~\cite{Wu2025SGM}. Self-Evolving
Agents (SEA) uses anytime-valid certificates to gate changes to a steering
adapter and a versioned harness around a frozen base model~\cite{Sengupta2026SEA}.
These works establish that certification and cumulative error control are
relevant to self-modifying AI. Our object of revision is a temporal
specification rather than an agent implementation; origin-version
obligations and the protected-trigger conditional target require separate
semantics and evidence.

\paragraph{Robust probabilistic shielding.}
Recent shielding methods combine learned uncertainty sets with formal
probabilistic guarantees. Galesloot, Rhemrev, and Jansen construct shields
for safe offline reinforcement learning from a fixed data set and obtain a
high-confidence policy-safety guarantee~\cite{Galesloot2026RobustOfflineShielding}.
Hamel-De le Court et al. synthesize sound and optimal shields for robust
MDPs, enforcing an LTL probability threshold against worst-case transition
probabilities and combining the construction with PAC transition
learning~\cite{HamelDeLeCourt2026RobustShielding}. These results already
cover probabilistic automata, data-derived model sets, and uniform robust
safety over the admissible dynamics. Our problem is different: the object
being changed is the temporal specification itself, rather than the policy
or its permitted actions. Consequently, the formal obligations are to keep
pending obligations attached to their origin version, preserve a
designer-fixed syntactic core, and certify both the aggregate and the
protected-trigger conditional targets of an adaptively selected revision.
The present lifetime bound ranges over specification activations; it is not
a replacement for model-robust shielding.

\paragraph{Monitoring deployed models and subpopulation guarantees.}
A related line monitors a deployed model sequentially and separates shifts
that degrade performance from benign ones, either by sequential testing on a
tracked risk~\cite{PodkopaevRamdas2022} or by weighted conformal test
martingales that adapt to mild shifts while detecting and diagnosing harmful
ones~\cite{Prinster2025WATCH}. Our aggregate versus protected-group
decomposition is related but not learned: the partition is the
designer-fixed protected trigger, it is immutable under revision, and it is
the same predicate that the deterministic invariant protects. Certifying a
fixed subpopulation separately from the aggregate is likewise studied
outside runtime verification~\cite{CherianCandes2024}; what is specific to
our setting is its interaction with adaptively selected candidates and with
repeated activations under a single lifetime budget.

Our distinguishing principle is
\[
\boxed{\text{AI proposal}+\text{symbolic admissibility}+\text{statistical certification}+\text{controlled activation}.}
\]
The AI component proposes candidate revisions but is not part of the trusted computing base. The symbolic governor decides activation and maintains guarantees across an adaptively generated sequence of specification versions.

We do not claim a new concentration inequality, a new temporal-logic
learning algorithm, or a new shielding construction. The contribution is
the certification object induced by adaptive specification revision: a
candidate chosen from the observed prefix is monitored prospectively before
activation, its completed outcomes are indexed by an outcome filtration,
and pending operational obligations remain attached to their origin
versions. Within that semantics, deterministic preservation of the
designer-fixed core and simultaneous aggregate/core evidence compose into a
lifetime guarantee over actual specification activations.

\section{Model and Temporal Semantics}

\subsection{Timed streams}

Let $\AP$ be a finite set of atomic propositions. A timed observation is a pair
$
e_n=(t_n,\nu_n),
$
where $n\in\mathbb N$ is the event index,
$t_n\in\mathbb R_{\ge0}$ is a timestamp, and
$\nu_n\subseteq\AP$ is the set of atomic propositions observed at time $t_n$.

\begin{definition}[Timed stream]
A timed stream is a finite or infinite sequence
$
\mathbf{e}=e_0e_1e_2\cdots
$
of timed observations with non-decreasing timestamps. For a finite stream
$\mathbf{e}=e_0\cdots e_{N-1}$, we require
$
t_n\le t_{n+1}$ and $n=0,\ldots,N-2,
$
whereas for an infinite stream the same condition holds for every
$n\in\mathbb N$.
In addition, every infinite timed stream is required to be non-Zeno:
$
\lim_{n\to\infty} t_n = \infty.
$
Thus every finite temporal horizon is eventually exceeded at a finite event
index.
\end{definition}

We write $\mathcal{E}^\omega$ for the set of infinite timed streams and
$\mathcal{E}^\ast$ for the set of finite timed streams. For an infinite stream
$\mathbf{e}\in\mathcal{E}^\omega$, the finite prefix available at event index $n$ is
$
\mathbf{e}_{\le n}=e_0e_1\cdots e_n\in\mathcal{E}^\ast.
$
We use $n$ exclusively for event indices.

\subsection{Temporal language}

Following the bounded-time monitoring tradition of metric temporal
logics~\cite{AlurHenzinger1993,MalerNickovic2004,BauerLeuckerSchallhart2011}, we consider the bounded
future-time fragment of MTL generated by
\[
\varphi ::= p
\mid \top \mid \bot
\mid \neg\varphi
\mid \varphi\land\varphi
\mid \varphi\lor\varphi
\mid \Box_{[a,b]}\varphi
\mid \Diamond_{[a,b]}\varphi,
\]
where $p\in\AP$ and
$0\le a\le b<\infty$.
We denote by $\Phi$ the set of formulas generated by this grammar.

The satisfaction relation is defined over infinite timed streams
$\mathbf{e}\in\mathcal{E}^\omega$ and event indices $n\in\mathbb N$.
The constants satisfy $\mathbf e,n\models\top$ and
$\mathbf e,n\not\models\bot$ for every stream and event index. Moreover
\begin{align*}
\mathbf{e},n\models p
&\iff p\in\nu_n,\\
\mathbf{e},n\models\neg\varphi
&\iff \mathbf{e},n\not\models\varphi,\\
\mathbf{e},n\models\varphi_1\land\varphi_2
&\iff
\mathbf{e},n\models\varphi_1
\text{ and }
\mathbf{e},n\models\varphi_2,\\
\mathbf{e},n\models\varphi_1\lor\varphi_2
&\iff
\mathbf{e},n\models\varphi_1
\text{ or }
\mathbf{e},n\models\varphi_2,\\
\mathbf{e},n\models\Box_{[a,b]}\varphi
&\iff
\forall m\ge n:
\bigl(t_m-t_n\in[a,b]\bigr)
\text{ implies }
\mathbf{e},m\models\varphi,\\
\mathbf{e},n\models\Diamond_{[a,b]}\varphi
&\iff
\exists m\ge n:
\bigl(t_m-t_n\in[a,b]\bigr)
\land
\mathbf{e},m\models\varphi.
\end{align*}

Finite streams $\mathbf{e}_{\le n}\in\mathcal{E}^\ast$ represent the
prefixes available to the online monitor and are not assigned the above
two-valued satisfaction relation.

Since all temporal intervals are bounded, every formula
$\varphi\in\Phi$ has a finite structural horizon. Set $\operatorname{hzn}(\top)=\operatorname{hzn}(\bot)=0$.  Define
\begin{align*}
\operatorname{hzn}(p) &= 0,\\
\operatorname{hzn}(\neg\varphi)
&= \operatorname{hzn}(\varphi),\\
\operatorname{hzn}(\varphi_1\circ\varphi_2)
&=
\max\{
\operatorname{hzn}(\varphi_1),
\operatorname{hzn}(\varphi_2)
\},
\qquad
\circ\in\{\land,\lor\},\\
\operatorname{hzn}(\Box_{[a,b]}\varphi)
&=
b+\operatorname{hzn}(\varphi),\\
\operatorname{hzn}(\Diamond_{[a,b]}\varphi)
&=
b+\operatorname{hzn}(\varphi).
\end{align*}

Hence, for every $\varphi\in\Phi$, the truth value of
$\mathbf e,n\models\varphi$ is determined once the timed stream has
progressed beyond
$t_n+\operatorname{hzn}(\varphi)$.

\subsection{Probabilistic trigger--response semantics}
\label{sec:specification}

In addition to trace-level formulas defined above, we consider probabilistic
trigger--response specifications.

\paragraph{Admissible parameter envelope.}
The designer fixes a measurable admissible parameter envelope
\[
\mathcal A_{\mathrm{par}}
\subseteq
\left\{
(\lambda,a,b)\in[0,1]\times\mathbb R_{\ge0}^2
\;\middle|\;
0\le a\le b<\infty
\right\}.
\]
The set $\mathcal A_{\mathrm{par}}$ specifies the parameter
combinations permitted by the designer. Each admissible triple
$(\lambda,a,b)$ determines a minimum success threshold $\lambda$
and a bounded response window $[a,b]$. Every specification version
must use parameters belonging to $\mathcal A_{\mathrm{par}}$, and
parameter revisions may select only triples within this set.

The envelope is fixed independently of the AI proposer and remains
unchanged throughout specification evolution. No closure, convexity,
or ordering assumption is imposed on $\mathcal A_{\mathrm{par}}$.
When desired, the designer may additionally impose a
strict lower threshold
$
\lambda_{\min}
:=
\inf\{\lambda:(\lambda,a,b)\in\mathcal A_{\mathrm{par}}\}>0.
$

\begin{definition}[Probabilistic trigger--response specification]
\label{def:trigger-response-specification}
A probabilistic trigger--response specification is the tuple
\[
\psi
=
\bigl(
\varphi_1^{\mathrm{core}},
\varphi_1^{\mathrm{adapt}},
\varphi_2^{\mathrm{core}},
\varphi_2^{\mathrm{adapt}},
\lambda,a,b
\bigr),
\]
where
$
\varphi_1^{\mathrm{core}},
\varphi_1^{\mathrm{adapt}},
\varphi_2^{\mathrm{core}},
\varphi_2^{\mathrm{adapt}}
\in\Phi$ and $(\lambda,a,b)\in\mathcal A_{\mathrm{par}}$.
Define the derived trigger and response by
$\operatorname{trig}(\psi)
=
\varphi_1^{\mathrm{core}}\lor\varphi_1^{\mathrm{adapt}}$ and $\operatorname{resp}(\psi)
=
\varphi_2^{\mathrm{core}}\land\varphi_2^{\mathrm{adapt}}$.
We use the abbreviation
\[
\psi
=
\operatorname{trig}(\psi)
\Rightarrow_{\lambda,[a,b]}
\operatorname{resp}(\psi)
\]
whenever the component decomposition is not itself under discussion, and
write $\lambda_\psi=\lambda$.

A specification without a protected core is represented by the convention
$
\varphi_1^{\mathrm{core}}=\bot$ and $
\varphi_2^{\mathrm{core}}=\top.
$
Let $
\Psi
\subseteq
\Phi^4\times\mathcal A_{\mathrm{par}}
$
denote the set of all such tuples. Equip $\Phi$ with a fixed
$\sigma$-algebra $\Sigma_\Phi$ and $\mathcal A_{\mathrm{par}}$ with its
chosen measurable structure (for instance, the trace Borel $\sigma$-algebra
when the envelope is Borel). We equip
$\Psi$
with the corresponding product/subspace $\sigma$-algebra
$
\Sigma_\Psi
=
\left.
\left(
\Sigma_\Phi^{\otimes4}
\otimes
\Sigma_{\mathcal A_{\mathrm{par}}}
\right)
\right|_{\Psi}.
$

Equality in $\Psi$ is componentwise tuple equality. Consequently, two
specifications with semantically equivalent derived trigger and response but
different core/adaptive decompositions are distinct elements of $\Psi$.
\end{definition}

A trigger--response obligation for $\psi$ is generated at event
index $n$ if
$\mathbf e,n\models\operatorname{trig}(\psi)$.
The obligation is successful if there exists an event index
$m\ge n$ such that
$t_m-t_n\in[a,b]$
and
$\mathbf e,m\models\operatorname{resp}(\psi)$.
The success condition is defined with respect to the infinite stream
and need not be decidable when the obligation is generated.

Define the observation horizon by
\[
H(\psi)
=
\max\!\left\{
\operatorname{hzn}(\operatorname{trig}(\psi)),
\;b+\operatorname{hzn}(\operatorname{resp}(\psi))
\right\}.
\]

The response window ends at $t_n+b$. The truth of the trigger at $n$
is guaranteed to be observable once the observed stream has advanced
beyond
$t_n+\operatorname{hzn}(\operatorname{trig}(\psi))$.
The success or failure of the response condition is guaranteed to be
observable once the stream has advanced beyond
$t_n+b+\operatorname{hzn}(\operatorname{resp}(\psi))$.

We adopt the conservative convention that an obligation is
\emph{completed} at the first observed event whose timestamp exceeds
$t_n+H(\psi)$. At completion, both the truth of its trigger and, if
the trigger holds, its success or failure are observable from the
finite prefix. An outcome may become observable earlier, but is
recorded as completed only at this boundary.

Given a timed stream
$\mathbf{e}=e_0e_1e_2\cdots$,
$e_n=(t_n,\nu_n)$,
the self-evolving monitor maintains a sequence of specification versions
$
\psi_0,\psi_1,\psi_2,\ldots .
$
We use $j$ exclusively for specification-version indices.

Let
$0=n_0<n_1<n_2<\cdots$
be the specification activation indices, with $n_j\in\mathbb N$ denoting
the event index at which version $\psi_j$ becomes active. At event index
$n$, the active specification is the unique version $\psi_j$ such that
$n_j\le n<n_{j+1}$.
If the version sequence is finite and $\psi_j$ is its last version, we use
the convention $n_{j+1}=\infty$.

\begin{example}[Specification activation indices]
\label{ex:specification-activation-indices}

Consider the timed stream
$
\mathbf{e}
=
e_0e_1e_2e_3e_4e_5\cdots,
$
with
\[
e_0=(0,\nu_0),\;
e_1=(2,\nu_1),\;
e_2=(4,\nu_2),\;
e_3=(7,\nu_3),\;
e_4=(9,\nu_4),\;
e_5=(12,\nu_5).
\]
Suppose that the specification activation indices are
$n_0=0$, $n_1=3$
and
$n_2=5$.

Then $\psi_0$ is active at event indices $0,1,2$,
$\psi_1$ is active at event indices $3,4$, and $\psi_2$ becomes
active at event index $5$.
Thus, $j$ indexes specification versions, whereas $n_j$ is an event index
identifying where the corresponding specification version becomes active
on the timed stream. In particular,
$e_5=(12,\nu_5)$,
so $\psi_2$ becomes active at event index $5$, corresponding to timestamp
$t_5=12$, immediately before $e_5$ is processed.
\end{example}

The evolution of the specification sequence is governed by an admissible
revision relation and an AI proposer.

\paragraph{Symbolic governor}
The admissible revision relation
$
\mathcal R\subseteq\Psi\times\Psi
$
is defined a priori by the system designer and specifies which structural
changes between specification versions are permitted.

\paragraph{Untrusted AI proposer}
Let $\Xi$ denote the space of random seeds used by the proposer.
At runtime, the possibly randomized proposal mechanism is represented by
\[
\mathsf{Propose}:
\mathcal{E}^\ast\times\Psi\times\Xi
\rightarrow
2^\Psi.
\]
Its internal randomness is included in the selection-time information
introduced in Section~\ref{sec:certification}.

When $\psi_j$ is active at event index $n$, the proposer computes
from the observed prefix $\mathbf{e}_{\le n}$ the set of candidate revisions
\[
\mathcal P_{n,j}
=
\mathsf{Propose}
(\mathbf{e}_{\le n},\psi_j,\xi_n^{\mathrm{prop}})
\subseteq\Psi,
\]
where $\xi_n^{\mathrm{prop}}\in\Xi$ is the random seed used
at that invocation.

The proposer returns decomposed specification tuples in $\Psi$.
The governor, rather than the proposer, enforces preservation of the
incumbent's protected components through $\mathcal R$.
Hence, $\mathcal R$ is a fixed design constraint, whereas
$\mathcal P_{n,j}$ is a data-dependent candidate set generated online.

\paragraph{Specification evolution.}

We distinguish two transition notions.

\begin{definition}[Specification evolution transitions]
\label{def:evolution-transitions}
For an active specification $\psi_j$ at event index $n$, we write
\[
\psi_j
\xrightarrow[n]{\mathcal R,\mathcal P}
\psi
\qquad\text{iff}\qquad
\psi\in\mathcal P_{n,j}
\quad\text{and}\quad
\psi_j\mathrel{\mathcal R}\psi.
\]
The relation
$\psi_j\xrightarrow[n]{\mathcal R,\mathcal P}\psi$
means that $\psi$ is proposed at event index $n$ and is an admissible
revision of $\psi_j$. This transition does not modify the active
specification.

We write
\[
\psi_j
\xRightarrow[n_{j+1}]{}
\psi_{j+1}
\]
to denote the subsequent certified activation of the next specification
version at event index $n_{j+1}$. We refer to this second transition as a
\emph{specification-version transition}. Its formal certification condition
is defined later in Section~\ref{sec:certification}.
\end{definition}

The certified-activation rule of Section~\ref{sec:certification} implies the
following fact:
\[
\psi_j
\xRightarrow[n_{j+1}]{}
\psi_{j+1}
\qquad\text{implies}\qquad
n_j\le n<n_{j+1}
\quad\text{and}\quad
\psi_j
\xrightarrow[n]{\mathcal R,\mathcal P}
\psi_{j+1}.
\]

\begin{example}[Proposed revision and certified activation]
\label{ex:proposed-certified-transition}

Consider again the timed stream and specification activation indices of
Example~\ref{ex:specification-activation-indices}. In particular,
$\psi_0$ is active at event indices $0,1,2$, and $\psi_1$ becomes active at
$n_1=3$.

A candidate that will eventually become $\psi_1$ may be proposed
and selected while $\psi_0$ is still active, for example at event
index $n=0$. Thus,
$\psi_0
\xrightarrow[0]{\mathcal R,\mathcal P}
\psi_1$
does not activate $\psi_1$. It only states that $\psi_1$ has
been proposed at event index $0$ and is an admissible revision
of $\psi_0$.

Suppose that sufficiently many post-selection obligations are
completed by event index $2$ and that the statistical
certification condition is satisfied.
The specification-version transition
$\psi_0
\xRightarrow[3]{}
\psi_1$
then occurs.
\end{example}

We adopt an origin-version semantics. A response obligation generated while
$\psi_j$ is active is evaluated according to the temporal interval and
trace-level formulas of $\psi_j$, even if $\psi_{j+1}$ becomes active before
the obligation is completed. Hence, a revision cannot retroactively change
the meaning of an already generated obligation.
More precisely, the operational version assigned to an origin
index $k$ is the unique $\psi_j$ satisfying $n_j\le k<n_{j+1}$.
This assignment is made at $k$, even if a future-time trigger cannot yet
be decided from $\mathbf e_{\le k}$: the monitor retains that origin and
its version until the trigger status and, when applicable, the obligation
outcome have been resolved. A candidate monitored before activation
produces certification-only observations, not operational obligations
under the candidate for those earlier origins.

The admissible revision relation $\mathcal R$ and the proposer-generated
candidate sets $\mathcal P_{n,j}$ are defined in detail in the following
sections. A proposed admissible candidate becomes the next active
specification only after satisfying the statistical certification conditions
defined in Section~\ref{sec:certification}.

\section{Controlled Specification Evolution: Admissible Revision Classes}
\label{sec:controlled-evolution}

Structural admissibility is determined by $\mathcal R$, whereas
statistical certification and activation are enforced separately
by the symbolic governor.
Consequently,
the structural guarantees below do not depend on the internal architecture,
training procedure, or correctness of the proposer.

The tuple representation of Definition~\ref{def:trigger-response-specification}
separates two designer-protected components from two adaptive components.
The protected components cannot be modified by specification evolution,
whereas either adaptive component may be revised. The disjunction in the
derived trigger ensures that every protected trigger remains covered, while
the conjunction in the derived response ensures that every successful full
response satisfies the protected response component.

\begin{definition}[Admissible revision]
\label{def:admissible-revision}
For any  $\psi,\widetilde\psi\in\Psi$ such that
$\psi\mathrel{\mathcal R}\widetilde\psi$
it holds that:

\begin{enumerate}
    \item the protected components of $\psi$ and $\widetilde\psi$ are unchanged;
    \item at most one of the two adaptive components of $\widetilde\psi$ differs from its
    counterpart in $\psi$;
    \item the parameter tuple of  $\widetilde\psi$ belongs to
    $\mathcal A_{\mathrm{par}}$; and
    \item at least one adaptive component or the parameter tuple of  $\widetilde\psi$ differs
    from its counterpart in $\psi$.
\end{enumerate}
Thus an admissible revision is a pure parameter shift, a trigger-adaptive
replacement, or a response-adaptive replacement, where an adaptive
replacement may be accompanied by a parameter shift.
\end{definition}

Structural admissibility does not impose a strengthening order between
successive versions. Within the designer-declared envelope, the threshold,
temporal interval, and one adaptive component may change in either direction.
Whether a structurally admissible candidate may become active is determined
separately by statistical certification.

\paragraph{Obtaining the decomposition.}
The deterministic invariant and the core-conditional certification of
Section~\ref{sec:certification} both rest on the protected components
$\varphi_1^{\mathrm{core}}$ and $\varphi_2^{\mathrm{core}}$ and on the
envelope $\mathcal A_{\mathrm{par}}$, all of which are supplied by the
designer before deployment and are never revised. The framework does not
infer them, and the guarantees it provides are relative to them: a vacuous
protected trigger ($\varphi_1^{\mathrm{core}}=\bot$) disables core
certification and reduces the invariant to the response core and the
envelope. Three considerations guide the choice in practice. First, the
protected trigger should capture the situations in which the requirement is
safety-relevant rather than the typical ones, since the invariant
guarantees coverage of exactly those; in the clinical-alarm setting of
Section~\ref{sec:introduction} this is the subset of desaturation events
that a clinical protocol classifies as requiring escalation, not all
alerts. Second, the protected response should retain only the conjuncts
whose violation is itself unsafe, because every additional conjunct is a
requirement that no future revision may relax, and an over-specified
response core makes admissible revisions impossible rather than unsafe.
Third, the protected share $w$ induced by
$\varphi_1^{\mathrm{core}}$ determines the cost of core certification:
as quantified in Remark~\ref{rem:core-certification-cost}, collecting a
fixed amount of core-conditional evidence requires approximately $1/w$
aggregate trigger observations. A narrowly defined protected trigger may
isolate the genuinely critical situations, but if those situations are
rare, substantially more aggregate observations are needed. This is an
inherent trade-off between the semantic scope of protection and
certification latency; the trigger should therefore be defined by the
safety requirement, not enlarged or narrowed merely to improve sample
efficiency.

\begin{theorem}[Deterministic protected-envelope invariant]
\label{thm:core-envelope-preserving}
Let $\psi_0
\mathrel{\mathcal R}
\psi_1
\mathrel{\mathcal R}
\cdots
\mathrel{\mathcal R}
\psi_q$
be any finite chain of admissible revisions, and write the protected
components of $\psi_0$ as
$\varphi_1^{\mathrm{core}}$ and
$\varphi_2^{\mathrm{core}}$.
Then, for every $i\in\{0,\ldots,q\}$:
\begin{enumerate}
    \item the protected components of $\psi_i$ are exactly
    $\varphi_1^{\mathrm{core}}$ and
    $\varphi_2^{\mathrm{core}}$;
    \item for every timed stream $\mathbf e$ and event index $n$,
    $\mathbf e,n\models\varphi_1^{\mathrm{core}}$
    implies
    $\mathbf e,n\models\operatorname{trig}(\psi_i)$;
    \item for every timed stream $\mathbf e$ and event index $m$,
    $\mathbf e,m\models\operatorname{resp}(\psi_i)$
    implies
    $\mathbf e,m\models\varphi_2^{\mathrm{core}}$;
    \item the parameter tuple of $\psi_i$ belongs to
    $\mathcal A_{\mathrm{par}}$.
\end{enumerate}
\end{theorem}

Theorem~\ref{thm:core-envelope-preserving} is a deterministic guarantee: it
holds for every admissible revision chain, rather than with probability
$1-\delta$. It is intentionally structural. In particular, the
trigger implication and response implication are pointwise properties;
they do not equate the temporal obligations of versions with different
response windows. Preserving the protected response as a conjunct does not
by itself guarantee that responses to protected triggers succeed with high
probability; the corresponding probabilistic requirement is enforced
separately by the core-conditional certification introduced in
Section~\ref{sec:certification}.

\section{Obligation Outcomes and Regime Semantics}
\label{sec:regime}

\subsection{Response outcomes and regime validity}

For the statistical statements in this section, fix a probability space
$(\Omega,\mathcal F,\mathbb P)$ on which the timed stream $\mathbf e$ is an
$\mathcal E^\omega$-valued random element. We equip $\mathcal E^\omega$ and the auxiliary spaces used by the proposal
and selection mechanisms with fixed $\sigma$-algebras, and equip $\Psi$
with the product/subspace $\sigma$-algebra $\Sigma_\Psi$ defined in
Section~\ref{sec:specification}. We assume
that the timed stream, proposer output, selected candidate, and all other
random objects entering the generated $\sigma$-algebras below are measurable
with respect to their corresponding measurable spaces.
 All response outcomes and
filtrations below are defined on this probability space.

Before defining the binary response outcome, we formalize when a
trigger--response obligation is generated and when its response window has
completed.

\begin{definition}[Completed trigger--response obligation]
\label{def:completed-obligation}
Let
$\psi
=
\varphi_1
\Rightarrow_{\lambda,[a,b]}
\varphi_2$.
\emph{A trigger--response obligation is generated at event index $n$} if
$\mathbf{e},n\models\varphi_1$.
The obligation generated at $n$ is \emph{completed} once its binary outcome
is guaranteed to be observable, namely if there exists an event index
$m\ge n$ such that
$t_m>t_n+H(\psi)$.
\end{definition}

Note that the response window $[t_n+a,t_n+b]$ may close before the obligation
is completed in the sense above, because evaluating $\varphi_2$ may itself
require future observations. As an example for $\varphi_2=\Diamond_{[0,5]}p$ it holds $\operatorname{hzn}(\varphi_2)=5$. Completion does not imply satisfaction;
satisfaction or violation is determined by the binary response variable
defined below. When $H(\psi)=0$, the strict inequality is conservative: it
waits until the stream progresses to a strictly later timestamp even though
a purely propositional zero-delay outcome may already be decidable at $n$.

For every event index $k$, define the binary response indicator
for a fixed version $\psi_j$ by
\[
X_k^{\psi_j}
=
\mathbf 1\!\left[
\exists m\ge k:
 t_m-t_k\in[a_j,b_j]
 \land
 \mathbf e,m\models\operatorname{resp}(\psi_j)
\right].
\]
When $\mathbf e,n\models\varphi_1$, this indicator is the binary outcome of
the obligation generated at $n$: $X_n^\psi=1$ denotes satisfaction and
$X_n^\psi=0$ violation. At runtime this outcome is admitted to monitoring or
certification only after the obligation is completed according to
Definition~\ref{def:completed-obligation}. Defining $X_n^\psi$ on the whole
probability space makes the trigger-conditioned probability used below
well defined. Each trigger generates a distinct obligation outcome even when
response windows overlap.

The following theorem is immediate once origin-version semantics is fixed;
its role is not technical depth but removal of a genuine ambiguity about
pending obligations when a revision occurs before their response windows
have completed.

\begin{theorem}[Origin-version non-retroactivity]
\label{thm:origin-nonretroactivity}
Let an obligation originate at index $k$ under $\psi_j$.
Under origin-version semantics, its trigger status, outcome
$X_k^{\psi_j}$, and completion index
\[
c_{\psi_j}(k)
:=
\min\{m\ge k:t_m>t_k+H(\psi_j)\}
\]
depend only on $(\mathbf e,k,\psi_j)$ and are unaffected by subsequent
revisions.

For every origin index $k$, there exists exactly one $j$ such that
$n_j\le k<n_{j+1}$.
Consequently, the obligation originating at $k$ is assigned
to $\psi_j$, independently of when it completes.
\end{theorem}

The origin-version choice is statistically consequential rather than merely
notational. To see this, consider instead a hypothetical \emph{active-version semantics} that reinterprets an unresolved obligation using the specification
active when its response is evaluated. Choose distinct $u,v,w\in\AP$, an
admissible parameter envelope containing $(1,1,1)$, and a specification
space $\Psi$ containing the tuples
$\psi=(u,\bot,\top,v,1,1,1)$
and
$\psi'=(u,\bot,\top,w,1,1,1)$.
Their derived specifications are
$u\Rightarrow_{1,[1,1]}v$ and
$u\Rightarrow_{1,[1,1]}w$, respectively. The revision
$\psi\mathrel{\mathcal R}\psi'$ changes only the adaptive response.
Let $t_k=0$, $t_{k+1}=1$, $t_{k+2}=2$, with $u$ true at $k$,
$v$ true and $w$ false at $k+1$. The obligation originates under
$\psi$ at $k$ and, by the strict completion convention, completes
at $k+2$. Under origin-version semantics it succeeds regardless of
whether $\psi'$ is activated at $k+1$. Under the hypothetical
active-version interpretation, the same unresolved obligation succeeds
if $\psi$ remains active, but fails if $\psi'$ is activated at
$k+1$ and its response is used at completion. Thus two activation
histories with identical timed observations and the same origin version
can produce different active-version outcomes. This counterexample
concerns the interpretation of a pending obligation, not whether the
hypothetical activation has passed a particular statistical test.

Consequently, Assumption~\ref{ass:regime-stability} below is naturally
well-posed under origin-version semantics because it indexes a success
probability by a fixed specification $\psi$. Under active-version semantics,
the analogous probability would additionally depend on the stochastic
revision process; a statistical treatment would therefore have to enlarge
the probabilistic model to include that process explicitly. This observation
does not rule out such a model, but it shows why the fixed-$\psi$ regime
assumption used in this paper relies on non-retroactivity.

\begin{definition}[Regime sequence]
\label{def:regime-sequence}
Behavioral evolution is represented by an ordered sequence of regimes
$R_0,R_1,R_2,\ldots$,
we use $r$ exclusively for regime indices.
Each $R_r\subseteq\mathbb N$ is a nonempty
contiguous set of event indices. Whenever $R_{r+1}$ exists, $R_r$ is finite
and
$\min \{R_{r+1}\}
=
\max \{R_r\}+1$.
The regime sequence is a model-provided segmentation of the event-index
axis and is not defined through the success probabilities of any particular
specification. Regimes therefore represent exogenous behavioral phases of
the underlying process.
\end{definition}

\begin{assumption}[Regime-level stability]
\label{ass:regime-stability}
Let $\Pr$ denote the probability law governing the random timed
stream $\mathbf e$.
For a specification
$\psi=\varphi_1\Rightarrow_{\lambda,[a,b]}\varphi_2$
and a regime $R_r$, there exists
$p_{R_r}(\psi)\in[0,1]$. For each $n\in R_r$, let
\[
E_{n,r}^{\psi}
:=
\left\{
\exists m\in R_r:
m\ge n\text{ and }t_m>t_n+H(\psi)
\right\}
\]
be the event that the obligation originating at $n$ can be completed within
$R_r$. For every $n\in R_r$ such that
$\Pr(\{\mathbf e,n\models\varphi_1\}\cap E_{n,r}^{\psi})>0$, we require
\[
\Pr\!\left(
X_n^\psi=1
\mid
\mathbf e,n\models\varphi_1,
E_{n,r}^{\psi}
\right)
=
p_{R_r}(\psi).
\]

This assumption is invoked only for the particular incumbent or
candidate specification whose regime-level interpretation is needed;
it is not a simultaneous stationarity assumption over all
$\psi\in\Psi$.
\end{assumption}

A regime-level success parameter is identified by
Assumption~\ref{ass:regime-stability} only when the regime contains at
least one origin index with positive probability of both triggering and
completion within the regime.
If no such index exists, the assumption is vacuous and no regime-validity
claim is made.
Whenever Assumption~\ref{ass:regime-stability} holds and such
an index exists, the specification
$\psi$ is valid in regime $R_r$ if
\[
R_r\models\psi
\quad\Longleftrightarrow\quad
p_{R_r}(\psi)\ge\lambda_\psi.
\]
The statistical certification results below do not require
Assumption~\ref{ass:regime-stability}; it is used only to interpret a
certified average predictable success probability as a regime-level success
probability.

\begin{example}[Regime and specification indices]
\label{ex:regime-specification-indices}

Consider again the timed stream and specification activation indices of
Example~\ref{ex:specification-activation-indices}. In particular,
$n_0=0$,
$n_1=3$ and
$n_2=5$.

Suppose that the behavioral regime sequence is
$R_0=\{0,1\}$,
$R_1=\{2,3,4\}$ and
$R_2=\{5,6,7,\ldots\}$.

Then the transition from $R_0$ to $R_1$ occurs at event index $2$, while
$\psi_0$ is still active. The specification-version transition from
$\psi_0$ to $\psi_1$ occurs only at event index $n_1=3$.

Hence, regime boundaries and specification activation boundaries need not
coincide. In particular, the same specification $\psi_0$ is active in both
$R_0$ and the initial part of $R_1$. It may therefore have different
regime probabilities
$p_{R_0}(\psi_0)\neq p_{R_1}(\psi_0)$.
\end{example}

\section{Statistical Certification and Activation}
\label{sec:certification}

\paragraph{Notational convention.}
From this section onwards a single specification-version transition is under
consideration and we suppress its index $j$ on all run-level quantities,
writing $n^\star$, $\psi^\star$, $\lambda^\star$ and $H^\star$ for the
selection index, the selected candidate, its threshold and its observation
horizon. Certification quantities then carry a single index: an event index
when written $m$, a certification sample size when written $q$ or $s$. In
particular $N_m$ is the sample size available at event index $m$, so that
$\bar p_{N_m}$ is the running average $\bar p_q$ evaluated at $q=N_m$, and
$K_m$, $L_m$, $\widehat p_{N_m}$ are the corresponding trigger set, lower
bound and empirical frequency. A superscript $\mathrm{core}$ marks the
protected-trigger subsample, and a group index
$g\in\{\mathrm{all},\mathrm{core}\}$ is used where the aggregate and core
statements have the same form. Objects belonging to the \emph{sequence} of
versions retain the index: $\psi_j$, $n_j$, $\xi_j$, $\mathcal H_j$,
$\mathcal G_{j,s}$, $\delta_j$, $m_j^{\mathrm{cert}}$ and $F_j$.

\subsection{Drift as a proposal gate}
\label{sec:drift}

The governor separates candidate selection from activation. A drift
declaration enables an evolution attempt but does not contribute to the
activation certificate. Structural admissibility is checked when the
candidate is selected, whereas activation depends only on the subsequent
certification evidence of Section~\ref{sec:certification}.

The prototype uses a deterministic two-window detector. For a fixed window
length $h$, it compares the empirical success frequencies in the two most
recent non-overlapping windows of $h$ completed incumbent obligations and
declares a decrease when their difference exceeds a fixed margin. The
end-to-end experiment of Section~\ref{sec:rq3-experiment} exercises this
detector. The masked-core experiment instead supplies the drift declaration
exogenously in order to isolate the certification mechanism. The detector
affects when candidate selection begins but not the coverage guarantee of
the subsequent certification procedure.

A prefix-measurable detector cannot establish the success probability of
obligations generated after activation: observationally identical prefixes
may admit different continuations. Interpreting the activation certificate
as a guarantee for such obligations therefore requires the additional
stability condition discussed in Section~\ref{sec:activation-scope}.

\subsection{Candidate selection}

Recall that the proposer is a possibly learned and randomized mechanism
represented by
\[
\mathsf{Propose}:
\mathcal{E}^\ast\times\Psi\times\Xi
\rightarrow
2^\Psi.
\]

At event index $n$, let $\psi_j$ be the active specification. The proposer
returns a set of candidate revisions
\[
\mathcal{P}_{n,j}
=
\mathsf{Propose}
(\mathbf{e}_{\le n},\psi_j,\xi_n^{\mathrm{prop}})
\subseteq\Psi,
\]
where $\xi_n^{\mathrm{prop}}\in\Xi$ is the random seed used at that
invocation.

The proposer may be implemented by any learned or data-driven model; its
internal parameters are outside the trusted computing base.
No correctness or soundness assumption is made on
$\mathsf{Propose}$. In particular, it may return statistically unsupported
candidates, candidates violating protected requirements, or no candidate at
all.

Recall also that
$
\psi_j
\xrightarrow[n]{\mathcal R,\mathcal P}
\psi
$
denotes a proposed admissible revision while $\psi_j$ remains active, and
$
\psi_j
\xRightarrow[n_{j+1}]{}
\psi_{j+1}
$
denotes the subsequent certified activation of the next specification
version.

The index $j$ labels an \emph{actually activated version}, not an attempt
or a detector alarm. 
While $\psi_j$ is active, a drift declaration only enables
the governor to consider a proposal; selecting an admissible
candidate starts its separate, prospective certification. Neither
step changes $\psi_j$. The active specification changes and $j$ advances
only after the candidate first meets the activation test defined below.
Without a selected and certified candidate there is no next version, and
$\psi_j$ remains active. The bounds in this section are therefore used to
control the decision to activate a replacement, not to presume that a
replacement must occur.

The proposer may return multiple candidates in $\mathcal P_{n,j}$.
Among them, the admissible candidates are those satisfying
$\psi_j
\xrightarrow[n]{\mathcal R,\mathcal P}
\psi$.

For each version $\psi_j$ that becomes active, let $n^\star$ be the
event index of the first (and only) candidate selected by the governor
for its possible successor. Set $n^\star=\infty$ if no candidate is
ever selected while $\psi_j$ is active. For version indices never
reached by the execution, use the same convention and leave the
corresponding candidate undefined. On $\{n^\star<\infty\}$, let
$\psi^\star$ be the selected candidate. 
Selection is enabled by a drift declaration and requires
$\psi_j
\xrightarrow[n^\star]{\mathcal R,\mathcal P}
\psi^\star$.
The drift declaration and structural admissibility enable
\emph{selection}, not activation. There is at most one
certification attempt per active version; an unsuccessful
attempt does not create $\psi_{j+1}$ and cannot be replaced
by a second candidate under the same $j$.

The selection may depend arbitrarily on the observed prefix
$\mathbf{e}_{\le n^\star}$, on the candidate set
$\mathcal P_{n^\star,j}$, and on internal randomness available by the
selection time, but not on the post-selection outcomes subsequently used for
certification.

\subsection{Selection-time information}

\begin{definition}[Selection-time information]
\label{def:selection-information}
Recall the extended selection index
$n^\star\in\mathbb N\cup\{\infty\}$.
On $\{n^\star<\infty\}$, let $\xi_j$ collect the internal randomness
used by the runtime decision mechanisms up to selection, including the
proposer seed $\xi_{n^\star}^{\mathrm{prop}}$ and, when applicable, the
detector seed $\zeta_{n^\star}$. On $\{n^\star=\infty\}$, assign fixed
cemetery values to the stopped prefix $\mathbf e_{\le n^\star}$, the
candidate set $\mathcal P_{n^\star,j}$, the candidate $\psi^\star$, and
$\xi_j$; these objects are not used for certification on that event.

Define
\[
\mathcal H_j
=
\sigma\!\left(
\mathbf 1_{\{n^\star<\infty\}},
n^\star,
\mathbf e_{\le n^\star},
\mathcal P_{n^\star,j},
\psi^\star,
\xi_j
\right).
\]
Thus the fact of selection, its index, the observed prefix, and all
proposal-and-selection information are $\mathcal H_j$-measurable. No
post-selection certification outcome is included on paths with finite
selection time.
\end{definition}

On $\{n^\star<\infty\}$, $\mathcal H_j$ represents the information
available when $\psi^\star$ is selected for certification. In particular,
it contains no certification outcome generated by a trigger occurring
after $n^\star$. All subsequent certification samples, filtrations, and
bounds are defined only on this event.

By construction,
$\sigma(\mathbf e_{\le n^\star})
\subseteq
\mathcal H_j$.
Thus every selection-time decision based on the observed prefix and the
internal randomness collected in $\xi_j$ is $\mathcal H_j$-measurable.

\subsection{Candidate certification samples}
\label{sec:certification-samples}

On $\{n^\star<\infty\}$, let
\[
\psi^\star
=
\varphi_1
\Rightarrow_{\lambda^\star,[a,b]}
\varphi_2
\]
be the candidate selected at event index $n^\star$, and write
$X_k:=X_k^{\psi^\star}$
for the binary outcome of the candidate obligation originating at event
index $k$.

\paragraph{Aggregate and protected-trigger samples.}
Two certification samples are used: the \emph{aggregate} sample of all
post-selection completed candidate obligations and, when a protected
trigger is present, the \emph{core} sample restricted to protected
triggers. Let
$g\in\{\mathrm{all},\mathrm{core}\}$
index the enabled samples. The aggregate sample is always enabled; when
the superscript is omitted, $g=\mathrm{all}$.

For every event index $m>n^\star$, define
\[
K_m^{(\mathrm{all})}
=
\left\{
k>n^\star:
\mathbf e,k\models\varphi_1,\;
t_m>t_k+H^\star
\right\},
\qquad
K_m:=K_m^{(\mathrm{all})}.
\]
Thus $K_m$ contains exactly the post-selection trigger indices whose
obligations under $\psi^\star$ are fully observable by event index $m$.

When $\varphi_1^{\mathrm{core}}\neq\bot$, define
\[
K_m^{(\mathrm{core})}
=
\left\{
k\in K_m:
\mathbf e,k\models\varphi_1^{\mathrm{core}}
\right\}.
\]
For every enabled group $g$, set
\[
N_m^{(g)}
:=
\left|K_m^{(g)}\right|,
\qquad
N_m:=N_m^{(\mathrm{all})}.
\]

In all cases,
$K_m
\subseteq
\{n^\star+1,\ldots,m-1\}$.
When the core sample is enabled,
$K_m^{(\mathrm{core})}\subseteq K_m$.

For every enabled group $g$, the certification sample available at event
index $m$ is
\[
\left(X_k\right)_{k\in K_m^{(g)}}.
\]
The core sample tests the \emph{full} candidate response
$\operatorname{resp}(\psi^\star)$ on protected triggers, not merely the
protected response component in isolation.

The trigger condition $\mathbf e,k\models\varphi_1$ is a trace-level
property and does not depend on which specification version is currently
active. Consequently, the selected candidate can be monitored
prospectively before activation. The condition $t_m>t_k+H^\star$ ensures
that the trigger status and complete response outcome of the obligation
originating at $k$ are observable by event index $m$.

The designer fixes $\lambda_{\mathrm{core}}\in[0,1]$ independently of the
proposer and of candidate selection. When
$\varphi_1^{\mathrm{core}}=\bot$, the core sample and its certification
condition are disabled.

\begin{assumption}[Post-selection sample recurrence]
\label{ass:trigger-recurrence}
For every selected candidate and every enabled group $g$,
$\bigcup_{m>n^\star}K_m^{(g)}$
is almost surely infinite.
\end{assumption}

Assumption~\ref{ass:trigger-recurrence} keeps the certification estimand
conditional on actual trigger obligations and makes the time-uniform
statement below well defined for every sample size $q\ge1$. A
finite-trigger extension can instead be obtained by conservative padding,
but it targets a different quantity because non-occurrence is then
incorporated into the estimand.

For each enabled group $g$, enumerate the corresponding post-selection
origins in increasing order:
\[
\bigcup_{m>n^\star}K_m^{(g)}
=
\left\{
k_1^{(g)},k_2^{(g)},\ldots
\right\},
\qquad
k_1^{(g)}<k_2^{(g)}<\cdots .
\]
Define the event index at which the $s$-th outcome of group $g$ becomes
available by
\[
\tau_s^{(g)}
=
\min\left\{
m>n^\star:
N_m^{(g)}\ge s
\right\},
\]
and define the associated completed-outcome filtration by
\[
\mathcal G_{j,0}^{(g)}
=
\mathcal H_j,
\qquad
\mathcal G_{j,s}^{(g)}
=
\sigma\!\left(
\mathcal H_j,
\left(
k_r^{(g)},
\tau_r^{(g)},
X_{k_r^{(g)}}
\right)_{r=1}^{s}
\right).
\]
When several outcomes become available at the same event index, they are
entered into the certification sequence in increasing order of their
origin indices. Thus each outcome is added in one filtration step, even
when consecutive completion indices coincide.

Both $k_s^{(g)}$ and $X_{k_s^{(g)}}$ are
$\mathcal G_{j,s}^{(g)}$-measurable. Moreover,
$\left\{
\tau_s^{(g)}\le m
\right\}
=
\left\{
N_m^{(g)}\ge s
\right\}$, for every event index $m$.
The latter event is determined by $\mathcal H_j$ and
$\mathbf e_{\le m}$; hence $\tau_s^{(g)}$ is a stopping index for the
post-selection observation filtration.

Because all obligations of the selected candidate share the observation
horizon $H^\star$ and timestamps are non-decreasing, completion preserves
the order of their origins. Therefore, for every $m>n^\star$,
\begin{equation}
\label{eq:initial-segment}
K_m^{(g)}
=
\left\{
k_1^{(g)},\ldots,
k_{N_m^{(g)}}^{(g)}
\right\}.
\end{equation}
Accordingly, the sample available at event index $m$ consists exactly of
$X_{k_1^{(g)}},\ldots,
X_{k_{N_m^{(g)}}^{(g)}}$.

\subsection{Anytime-valid certification}

\paragraph{Certification statistics.}
For every enabled group $g$ and every $s,q\ge1$, define
\[
p_s^{(g)}
=
\Pr\!\left(
X_{k_s^{(g)}}=1
\mid
\mathcal G_{j,s-1}^{(g)}
\right),
\]
\[
\bar p_q^{(g)}
=
\frac1q
\sum_{s=1}^{q}p_s^{(g)},
\qquad
\widehat p_q^{(g)}
=
\frac1q
\sum_{s=1}^{q}X_{k_s^{(g)}}.
\]

Because $X_{k_s^{(g)}}$ is binary,
\[
p_s^{(g)}
=
\mathbb E\!\left[
X_{k_s^{(g)}}
\mid
\mathcal G_{j,s-1}^{(g)}
\right]
\]
is also the predictable conditional success probability of the $s$-th
outcome of group $g$. The quantity $\bar p_q^{(g)}$ is the running average
of these predictable conditional probabilities, while
$\widehat p_q^{(g)}$ is the corresponding observed frequency.

By Equation~\eqref{eq:initial-segment},
\[
\widehat p_{N_m^{(g)}}^{(g)}
=
\frac{1}{N_m^{(g)}}
\sum_{k\in K_m^{(g)}}X_k
\]
whenever $N_m^{(g)}>0$. Hence
$\widehat p_{N_m^{(g)}}^{(g)}$ is precisely the observed success frequency
of the outcomes of group $g$ available at event index $m$.

No stationarity or independence assumption is imposed on
$(p_s^{(g)})_{s\ge1}$. In particular, a certification sample may contain
outcomes affected by a behavioural transition.

At selection time, fix $\mathcal H_j$-measurable error allocations
\[
\delta_j^{\mathrm{all}}>0,
\qquad
\delta_j^{\mathrm{core}}\ge0,
\qquad
\delta_j^{\mathrm{all}}
+
\delta_j^{\mathrm{core}}
\le
\delta_j.
\]
If the protected trigger is $\bot$, the core group is disabled and
$\delta_j^{\mathrm{core}}=0$; otherwise
$\delta_j^{\mathrm{core}}>0$. Conditioning on $\mathcal H_j$ therefore
fixes all enabled allocations before any post-selection outcome is
observed.

For every enabled group $g$ and every $q\ge1$, set
\[
\delta_{j,q}^{(g)}
=
\frac{6\,\delta_j^{(g)}}{\pi^2q^2},
\qquad
\sum_{q\ge1}\delta_{j,q}^{(g)}
=
\delta_j^{(g)}.
\]
For $N_m^{(g)}>0$, define
\[
L_m^{(g)}
=
\max\left\{
0,\;
\widehat p_{N_m^{(g)}}^{(g)}
-
\sqrt{
\frac{
\log\!\left(
1/\delta_{j,N_m^{(g)}}^{(g)}
\right)
}{
2N_m^{(g)}
}
}
\right\},
\]
and set $L_m^{(g)}=0$ when $N_m^{(g)}=0$.

The bound targets the average predictable success probability of the
outcomes actually used for certification rather than assuming a constant
regime probability. When the core group is disabled, only the aggregate
instance is defined and evaluated.

\begin{definition}[Statistical certification]
\label{def:statistical-certification}
The selected candidate $\psi^\star$ is statistically certified at event
index $m>n^\star$ if
$N_m>0$,
$L_m\ge\lambda^\star$,
and, whenever core-conditional certification is enabled,
$N_m^{\mathrm{core}}>0$,
$L_m^{\mathrm{core}}
\ge
\lambda_{\mathrm{core}}$.
\end{definition}

This predicate tests the selected candidate using completed
post-selection outcomes. Passing it is a prerequisite for activation;
evaluating it does not itself change the active specification.

\paragraph{Masked core failure: why core-conditional certification is needed.}
Fix a specification $\psi_j$ and a regime $R_r$ satisfying
Assumption~\ref{ass:regime-stability}. For each eligible origin index $k$,
protected triggers and adaptive-only triggers partition the event
$\mathbf e,k\models\operatorname{trig}(\psi_j)$ into the disjoint events
$\mathbf e,k\models\varphi_1^{\mathrm{core}}$ and
$\mathbf e,k\models
\varphi_1^{\mathrm{adapt}}\land
\neg\varphi_1^{\mathrm{core}}$.
Suppose both groups have positive probability, and their conditional
probabilities below are independent of the eligible origin index $k$
within $R_r$. Define
\[
\begin{aligned}
w_r
&:=\Pr(\mathbf e,k\models\varphi_1^{\mathrm{core}}
 \mid\mathbf e,k\models\operatorname{trig}(\psi_j),
 E_{k,r}^{\psi_j}),\\
p_r^{\mathrm{core}}
&:=\Pr(X_k^{\psi_j}=1
 \mid\mathbf e,k\models\varphi_1^{\mathrm{core}},
 E_{k,r}^{\psi_j}),\\
p_r^{\mathrm{adapt}}
&:=\Pr(X_k^{\psi_j}=1
 \mid\mathbf e,k\models
 \varphi_1^{\mathrm{adapt}}\land
 \neg\varphi_1^{\mathrm{core}},
 E_{k,r}^{\psi_j}).
\end{aligned}
\]
The conditional success probabilities concern the \emph{full} response,
not the protected response component alone. The within-group stability
assumption is additional to Assumption~\ref{ass:regime-stability}.
By the law of total probability, conditioning on the full trigger and
within-regime completion and then on this partition yields
\begin{equation}
\label{eq:core-mixture}
p_{R_r}(\psi_j)
=w_r p_r^{\mathrm{core}}
 +(1-w_r)p_r^{\mathrm{adapt}}.
\end{equation}
If only the protected group has positive probability, the identity
reduces to $p_{R_r}(\psi_j)=p_r^{\mathrm{core}}$; the adaptive-only
conditional probability is then left undefined. If the protected group
has zero probability, $p_r^{\mathrm{core}}$ is not identified from this
regime. In particular, when $\varphi_1^{\mathrm{core}}=\bot$,
core certification is disabled and only the aggregate quantity is used.
The same decomposition applies to a fixed candidate $\psi^\star$
when the corresponding positivity and within-group assumptions hold; we
denote its protected-group conditional probability in regime $R_r$ by
$p_{R_r}^{\mathrm{core}}(\psi^\star)$. In what follows,
\emph{aggregate-only certification} denotes the ablation of
Definition~\ref{def:statistical-certification} obtained by omitting its
core-conditional conjunct.
It does not by itself identify the predictable means of its
post-selection certification sample with a regime-level probability.

The following proposition shows that aggregate certification alone can
accept a candidate even when its full response fails systematically on
every protected trigger.

\begin{proposition}[Masked core failure]
\label{prop:masked-core-failure}
There exist a protected decomposition, two regimes
$R_0$ and $R_1$, and a candidate specification $\psi^\star$ such that
\[
p_{R_0}(\psi^\star)=p_{R_1}(\psi^\star),
\qquad
p_{R_0}^{\mathrm{core}}(\psi^\star)=0.99,
\qquad
p_{R_1}^{\mathrm{core}}(\psi^\star)=0.
\]
Moreover, if $R_1$ governs an infinite post-selection certification
sequence, aggregate-only certification eventually certifies
$\psi^\star$ almost surely, whereas the joint aggregate and
core-conditional criterion of
Definition~\ref{def:statistical-certification} never certifies it.
\end{proposition}

\begin{remark}[Core-sample acquisition cost]
\label{rem:core-certification-cost}
Fix
$\lambda_{\mathrm{core}}\in[0,1)$ and a core confidence allocation
$\delta_j^{\mathrm{core}}>0$. On a perfect-response stream, define
\[
N^\star
:=
\min\left\{
q\ge1:
1-
\sqrt{
\frac{
\log\!\bigl(\pi^2q^2/(6\delta_j^{\mathrm{core}})\bigr)
}{
2q
}
}
\ge
\lambda_{\mathrm{core}}
\right\}.
\]
Thus $N^\star$ is the minimum number of completed protected-trigger
obligations required for the core conjunct
$L_m^{\mathrm{core}}
\ge
\lambda_{\mathrm{core}}$
of Definition~\ref{def:statistical-certification} on a stream with no
observed core failures.

Suppose that, along the ordered aggregate certification sequence
$(k_s)_{s\ge1}$, protected-trigger membership is independent across
origins and has probability $w>0$. Define
\[
T_{N^\star}
:=
\min\left\{
q\ge1:
\sum_{s=1}^{q}
\mathbf 1\!\left[
\mathbf e,k_s
\models
\varphi_1^{\mathrm{core}}
\right]
\ge
N^\star
\right\}.
\]
Hence $T_{N^\star}$ is the aggregate sample size required to collect
$N^\star$ completed protected-trigger obligations. It has a
negative-binomial distribution, with
\[
\mathbb E[T_{N^\star}]
=
\frac{N^\star}{w},
\qquad
\operatorname{Var}(T_{N^\star})
=
\frac{N^\star(1-w)}{w^2}.
\]

Therefore, under the perfect-response benchmark and for fixed
$\lambda_{\mathrm{core}}$ and $\delta_j^{\mathrm{core}}$, the expected
number of completed aggregate obligations that must be examined before
the core sample contains $N^\star$ outcomes, and hence before
$L_m^{\mathrm{core}}\ge\lambda_{\mathrm{core}}$ can hold, scales as
$1/w$.
Thus the expected aggregate sample size is inversely proportional to the
protected-trigger probability $w$: halving $w$ doubles the expected number
of completed aggregate obligations required to collect the same core sample.

This is a sample-acquisition cost, not by itself the complete activation
latency. Statistical certification also requires the aggregate conjunct,
and converting a number of completed obligations into event time depends
on trigger arrivals and completion delays.

For the Hoeffding bound used here, if
$\lambda_{\mathrm{core}}=0.9$ and the transition budget
$\delta_0=0.025$ is split equally, then
$\delta_0^{\mathrm{core}}=0.0125$ and a perfect-response stream gives
$N^\star=928$. 
With a protected-trigger share of $w=0.01$, the expected number of
completed aggregate obligations required to collect the $928$ protected
outcomes is
$\mathbb E[T_{N^\star}]
=
\frac{928}{0.01}
=
92{,}800$.
This calculation quantifies the additional evidence-collection latency
introduced by core-conditional certification and makes explicit its main
trade-off: stronger protection of a rare subgroup requires a larger
aggregate sample.
\end{remark}

\paragraph{Choice of lower bound.}
The subsequent results do not depend on the specific Hoeffding formula;
they require only that the lower bound cover
$\bar p_q^{(g)}$ simultaneously for all sample sizes $q$, conditionally
on the selection-time information $\mathcal H_j$. We use the stated
Hoeffding bound~\cite{Hoeffding1963} because it admits the short
self-contained proof given below.

The bound is conservative. In the controlled i.i.d.\ Bernoulli experiment
of Section~\ref{sec:rq1-experiment}, at true success probability
$0.98$, it required
a median of $1{,}461$ completed obligations, compared with $240$ for an
exact binomial test using the same spending schedule. The exact test,
however, relies on the additional i.i.d.\ assumption and is not valid under
the general history-dependent model used in the formal results. Any sharper
lower bound satisfying the same simultaneous conditional coverage property
could be substituted without changing the remaining arguments
~\cite{WaudbySmithRamdas2024}.

The next proposition establishes that each enabled lower bound covers its
corresponding average predictable success probability simultaneously at
every certification index.

\begin{proposition}[Anytime-valid lower confidence bound]
\label{prop:anytime-lower-bound}
For each enabled group $g\in\{\mathrm{all},\mathrm{core}\}$, conditionally on $\mathcal H_j$,
\[
\Pr\!\left(\forall m>n^\star\text{ with }N_m^{(g)}>0:\;
\bar p_{N_m^{(g)}}^{(g)}\ge L_m^{(g)}
\;\middle|\;\mathcal H_j\right)\ge1-\delta_j^{(g)}.
\]
The core instance is omitted when core certification is disabled.
\end{proposition}

Applying the preceding result to both enabled samples yields their
simultaneous coverage without requiring independence between them.

\begin{proposition}[Simultaneous aggregate and core-conditional coverage]
\label{prop:joint-core-coverage}
Conditionally on $\mathcal H_j$,
\[
\Pr\!\left(
\begin{array}{c}
\forall\, g\in\{\mathrm{all},\mathrm{core}\}
\text{ enabled},\
\forall\, m>n^\star
\text{ with }N_m^{(g)}>0:\\[1mm]
\bar p_{N_m^{(g)}}^{(g)}
\ge
L_m^{(g)}
\end{array}
\;\middle|\;
\mathcal H_j
\right)
\ge
1-\delta_j.
\]
The core instance is omitted when core certification is disabled.
No independence between the aggregate and core samples is required.
\end{proposition}

Conditioning on $\mathcal H_j$ fixes the selected candidate and its
protected-trigger predicate, but not the subsequent trigger origins or
response outcomes. The coverage results therefore remain valid after
adaptive candidate selection.

They control the average predictable success probability
$\bar p_q^{(g)}$ of the certification sequence. They do not by themselves
establish a fixed regime-level probability or the success probability of
obligations generated after activation; those interpretations require the
additional alignment or stability conditions stated below.

The following corollary identifies the additional alignment condition under
which simultaneous coverage of the certification samples admits a
regime-level interpretation.

\begin{corollary}[Joint regime-level interpretation]
\label{cor:joint-regime-level-certification}
Fix an event index $m$ and a regime $R_r$. For the aggregate group, define
$p_{R_r}^{(\mathrm{all})}(\psi^\star)
:=
p_{R_r}(\psi^\star)$
and
$\lambda_{\mathrm{all}}
:=
\lambda^\star$.
For the core group, when enabled, retain
$p_{R_r}^{\mathrm{core}}(\psi^\star)$ and
$\lambda_{\mathrm{core}}$. Assume that
$p_{R_r}^{(g)}(\psi^\star)$ is well defined for every enabled group $g$.

Suppose that, for every enabled group
$g\in\{\mathrm{all},\mathrm{core}\}$ and every $s\ge1$,
\[
\mathbf 1_{\{n^\star<m,\;N_m^{(g)}\ge s\}}
p_s^{(g)}
=
\mathbf 1_{\{n^\star<m,\;N_m^{(g)}\ge s\}}
p_{R_r}^{(g)}(\psi^\star)
\qquad\text{a.s.}
\]
Then, conditionally on $\mathcal H_j$,
\[
\Pr\!\left(
\begin{array}{c}
n^\star<m,\;
\psi^\star\text{ is statistically certified at }m,\\
\exists\,g\in\{\mathrm{all},\mathrm{core}\}\text{ enabled}:
p_{R_r}^{(g)}(\psi^\star)<\lambda_g
\end{array}
\;\middle|\;
\mathcal H_j
\right)
\le
\delta_j.
\]
\end{corollary}

\subsection{Scope of the activation certificate}
\label{sec:activation-scope}

The certificate controls
$\bar p_{N_m^{(g)}}^{(g)}$, the average predictable success probability
of the outcomes used at the certification index. A regime-level
interpretation requires the alignment condition of
Corollary~\ref{cor:joint-regime-level-certification}.

The first obligation generated under the successor version is not the next
member of the certification sequence. Therefore, the certificate does not
by itself control the success probability of obligations generated after
activation. Such a conclusion requires an additional stability assumption.
Indeed, two probability laws may agree on the entire observed prefix,
including the certification samples, lower bounds, and activation decision,
while assigning different success probabilities to subsequent obligations.

For clarity, the remaining scope restrictions are collected here. Group
membership must be compatible with the completed-outcome filtration; fixing
a terminal-window length alone does not make a split predictable. Aggregate
and core samples may overlap, and their coverage failures are combined by a
union bound rather than an independence assumption. The present protocol
allows one selected candidate per active version; repeated attempts require
an additional predictable error allocation. Certification supplies no liveness guarantee and no future operational
guarantee without an additional persistence assumption relating the
certification regime to obligations generated after activation.

A model-based or hybrid core certificate is admissible if it provides a
valid lower bound on the protected-trigger conditional success probability
and accounts for its estimation and repeated-use error within the lifetime
budget. Aggregate model accuracy alone is insufficient: conditioning on a
protected event of probability $w$ may amplify a joint-law error of order
$\varepsilon_S$ to a conditional error of order
$\varepsilon_S/w$.

More generally, any additional outcome group must be defined so that its
certification sequence satisfies the filtration conditions used by the
confidence bound.

\subsection{Certified activation}

The end-to-end architecture uses a drift declaration as a
necessity gate for candidate selection, as formalized in
Section~\ref{sec:drift}. It is not a second condition checked at
activation: structural admissibility is verified at selection, and once
a candidate has been selected its activation depends on the
post-selection certification evidence.

On $\{n^\star<\infty\}$, the governor monitors the fixed, structurally
admissible candidate $\psi^\star$ without replacing the incumbent.
Define the first eligible certification index by
\[
\begin{aligned}
m_j^{\mathrm{cert}}
:=\inf\bigl\{m\in\mathbb N:\;&m>n^\star,
\ N_m>0,
L_m
\ge\lambda^\star,\\
&\ 
\left(
\varphi_1^{\mathrm{core}}=\bot
\;\lor\;
\left(
N_m^{\mathrm{core}}>0
\;\land\;
L_m^{\mathrm{core}}
\ge\lambda_{\mathrm{core}}
\right)
\right)\},
\end{aligned}
\]
with $\inf\varnothing=\infty$. If $n^\star=\infty$, set
$m_j^{\mathrm{cert}}=\infty$. The core test in this definition retains
both branches above: it is automatically true only for a specification
without a protected trigger.

An actual specification change occurs \emph{if and only if}
$m_j^{\mathrm{cert}}<\infty$. On this event define
\[
\psi_{j+1}:=\psi^\star,
\qquad
n_{j+1}:=m_j^{\mathrm{cert}}+1,
\qquad
\psi_j\xRightarrow[n_{j+1}]{}\psi_{j+1}.
\]
If $m_j^{\mathrm{cert}}=\infty$, there is no successor version,
$n_{j+1}=\infty$, and the index $j$ does not advance.
The certification decision is made after processing
$\mathbf e_{\le m_j^{\mathrm{cert}}}$; the new version first governs
$e_{m_j^{\mathrm{cert}}+1}$. Hence the incumbent $\psi_j$ remains active
throughout the selection and certification phases, and the first
successful certification test---not the drift alarm or proposal---activates
the modification.

\subsection{Lifetime guarantee}

For each possible specification-version transition $j$, let the total error
level $\delta_j\in(0,1)$ be chosen no later than the selection time represented
by $\mathcal H_j$. Thus $\delta_j$ and its split
$\delta_j^{\mathrm{all}},\delta_j^{\mathrm{core}}$ are deterministic or
$\mathcal H_j$-measurable and cannot depend on certification outcomes
revealed after selection. Assume that, almost surely,
\[
\sum_{j=0}^{\infty}\delta_j\le\delta
\]
for a fixed lifetime budget $\delta\in(0,1)$.

Only a finite $m_j^{\mathrm{cert}}$ creates $\psi_{j+1}$.
For a version index never reached, or when no candidate is selected or
certified, all failure events below are empty. On a run where activation
occurs, define
\[
F_j^{\mathrm{all}}
=
\left\{
\psi^\star\text{ is activated and }
\bar p_{N_{n_{j+1}-1}}
<\lambda^\star
\right\}.
\]
When core-conditional certification is enabled, define
\[
F_j^{\mathrm{core}}
=
\left\{
\psi^\star\text{ is activated and }
\bar p^{\mathrm{core}}_{N_{n_{j+1}-1}^{\mathrm{core}}}
<\lambda_{\mathrm{core}}
\right\};
\]
otherwise set $F_j^{\mathrm{core}}=\varnothing$. Set
$F_j=F_j^{\mathrm{all}}\cup F_j^{\mathrm{core}}$, and $F_j=\varnothing$ whenever $\psi_j$ is not
reached or no successor is activated.
Thus $F_j$ concerns only \emph{activated} candidates whose
aggregate or (when enabled) core predictable-mean requirement fails
at the activation decision. It does not assert regime validity at
activation, success of future obligations, or eventual adaptation.

The per-transition coverage guarantee extends to the entire sequence of
activated versions by allocating a summable error budget across version
transitions.

\begin{theorem}[Proposer-independent lifetime certification]
\label{thm:lifetime-certification}
For each version index $j\ge0$, suppose that at most one candidate is
selected while $\psi_j$ is active and that, whenever selection occurs, the
conditions of Proposition~\ref{prop:joint-core-coverage} hold with error
budget $\delta_j$. Suppose also that each $\delta_j$ is fixed no later than
selection time and that
$\sum_{j=0}^{\infty}\delta_j\le\delta$
almost surely
for some $\delta\in(0,1)$. Then
\[
\Pr\!\left(
\bigcup_{j=0}^{\infty}F_j
\right)
\le
\delta.
\]
\end{theorem}

The lifetime bound is independent of the architecture, training procedure,
and randomization of the proposer, provided that the selected candidate is
$\mathcal H_j$-measurable and is fixed before any post-selection
certification outcome is observed. In particular, the protected-trigger
predicate defining the core sample cannot be modified during certification.

Candidate multiplicity, sample overlap, repeated attempts, and liveness are
subject to the restrictions collected in
Section~\ref{sec:activation-scope}. The theorem itself controls only the
unconditional probability of an erroneous activated revision at its
decision index.
\section{End-to-End Architecture and Algorithm}
\label{sec:algorithm}
\label{sec:end-to-end}

The preceding sections define admissible revisions, candidate selection,
post-selection certification, and certified activation. We now state how
these components are executed at each event index. The protocol uses the
completed-obligation interface and does not assume any transfer of internal
monitor states between specification versions.

\paragraph{Monitoring across specification versions.}
For every origin index $k$, there is a unique specification version
$\psi_i$ such that
$n_i\le k<n_{i+1}$.
The trigger status, response outcome, and completion index of the obligation
originating at $k$ are evaluated using this origin version $\psi_i$.
Consequently, after $\psi_{i+1}$ becomes active, $\psi_i$ remains responsible
for completing obligations whose origins precede $n_{i+1}$, but it generates
no obligations at later origins.

This distinction is necessary because an obligation may be completed after
its origin version has ceased to be active. In particular, a trigger or
response containing future-time operators may become decidable only after a
later version has been activated. The subsequent activation does not alter
the trigger, response, interval, or outcome attached to the earlier origin.

A selected candidate $\psi^\star$ is monitored separately before
activation. Its certification-only monitor evaluates candidate obligations
with origins $k>n^\star$ and supplies their completed outcomes to the
aggregate and, when enabled, core certification samples. These outcomes are
certification evidence. They are not reclassified as operational
obligations generated under $\psi^\star$ if the candidate is later
activated.

\begin{figure}[t]
\centering
\caption{Ordered governor protocol at event index $n$.}
\Description{An ordered phase table showing how origin-version obligations,
candidate certification, activation, and proposal generation are processed
at each event index.}
\label{fig:governor-protocol}
\small
\begin{tabular}{@{}p{0.07\linewidth}p{0.18\linewidth}
                    p{0.27\linewidth}p{0.40\linewidth}@{}}
\toprule
Step & Phase & Condition & Governor action \\
\midrule

1
& Observe
& Every event $e_n$
& Advance all unresolved origin-version obligations. Record every trigger
status and response outcome that becomes observable at $n$, using the
specification version under which the corresponding obligation originated.
Update the detector for the active specification $\psi_j$. \\

2
& Certify
& A candidate $\psi^\star$ has been selected
& Advance the certification-only monitor on candidate origins
$k>n^\star$. Add newly completed outcomes to the aggregate sample and,
when enabled, to the core sample. Evaluate $L_n$ and
$L_n^{\mathrm{core}}$. \\

3
& Activate
& $\psi^\star$ satisfies
Definition~\ref{def:statistical-certification} for the first time at $n$
& Set
$m_j^{\mathrm{cert}}=n$,
$\psi_{j+1}=\psi^\star$, and
$n_{j+1}=n+1$. The new version governs origins from $n+1$ onward.
Obligations with earlier origins retain their origin versions. \\

4
& Propose
& a drift declaration occurs at $n$ and no candidate has yet been selected
while $\psi_j$ is active
& Obtain $\mathcal P_{n,j}$, discard every candidate that fails
$\psi_j\mathrel{\mathcal R}\psi$, select at most one survivor, fix its
error allocation, set $n^\star=n$, and start prospective certification. \\

5
& Continue
& Neither activation nor candidate selection occurs
& Keep $\psi_j$ active. Continue resolving obligations under their origin
versions. A drift declaration, proposal, or admissibility check alone does
not change the active specification. \\

\bottomrule
\end{tabular}
\end{figure}

\paragraph{Execution order.}
The rows of Figure~\ref{fig:governor-protocol} are evaluated in order. First,
the governor processes $e_n$ and resolves every obligation outcome that
becomes observable at that event index. If a candidate has already been
selected, these newly available outcomes may enlarge its certification
samples and change its lower bounds.

Activation is then checked before a new proposal is considered. If the
candidate first satisfies the certification condition at $n$, the
certification decision uses only the prefix $\mathbf e_{\le n}$ and the
successor version first governs event $e_{n+1}$. Thus
$n_{j+1}=m_j^{\mathrm{cert}}+1$.
No obligation with origin at or before $m_j^{\mathrm{cert}}$ is retroactively
assigned to $\psi_{j+1}$.

If no candidate has been selected and no activation occurs, a drift
declaration may enable proposal generation. Selection fixes
$\psi^\star$, $n^\star$, and its error allocation using information
available at that time. Certification then uses only completed obligations
whose origins satisfy $k>n^\star$. Proposal and selection do not alter
$\psi_j$.

The three indices therefore have distinct meanings:
$n^\star$
selects the candidate,
$m_j^{\mathrm{cert}}$
records the certification decision, $n_{j+1}$ starts the successor version.

\paragraph{Guarantees enforced by the protocol.}
Every activated successor satisfies the admissible-revision relation.
Therefore, the protected trigger and response components are unchanged and
the parameter tuple remains inside the designer-fixed envelope throughout
the activated version sequence.

The certification samples are separated from the prefix used to select the
candidate. At activation, the aggregate lower bound controls the average
predictable success probability of the completed aggregate outcomes used by
the decision. When core certification is enabled, the core lower bound
simultaneously controls the corresponding average on protected triggers.
The aggregate and core samples may overlap; their failure probabilities are
combined through their error allocations, without an independence
assumption.

For transition $j$, the allocation $\delta_j$ and its aggregate/core split
are fixed at selection time. For example,
$\delta_j=\frac{\delta}{2^{j+1}}$
gives
$\sum_{j=0}^{\infty}\delta_j\le\delta$.
Consequently, over the sequence of activated versions, the probability that
some activated candidate fails its aggregate or enabled core
predictable-mean requirement at its activation decision is at most
$\delta$.

This guarantee concerns the completed certification outcomes available at
the decision index. Obligations generated after activation are not the next
members of the certification sequence. Interpreting the certificate as a
statement about a regime or about successor-version operational obligations
requires the additional alignment condition stated in
Corollary~\ref{cor:joint-regime-level-certification}.

\paragraph{Runtime overhead and liveness.}
While a candidate is being certified, the active specification and the
selected candidate are monitored in parallel. In addition, earlier
specification versions may still have unresolved origin-version
obligations. Once all obligations attributed to an earlier version have
been resolved, no further monitoring work is required for that version.

The detector continues to update while certification is active, but further
declarations do not select another candidate. Under the
single-attempt-per-version rule, a selected candidate that never certifies
may therefore prevent further adaptation while the incumbent remains
active. Repeated attempts would require an additional attempt index and a
predictable allocation of $\delta_j$ across attempts.
\section{Experimental evaluation}
\label{sec:evaluation}

The evaluation contains three controlled synthetic experiments and a
monitored case study, each tied to a stated result. RQ1 illustrates the optional-continuation protection
of Proposition~\ref{prop:anytime-lower-bound} and quantifies the conservativeness
of the Hoeffding bound. RQ2 illustrates the masked-core counterexample
of Proposition~\ref{prop:masked-core-failure} at a single activation
decision. RQ3 exercises repeated proposal and certification under the
lifetime schedule of Theorem~\ref{thm:lifetime-certification}.

\paragraph{The proposer.}
All three experiments and the case study use the same proposer: a supervised
regressor trained offline on 4,000 independent synthetic tasks and frozen
before any experiment is run. It receives the prefix aggregate frequency,
the trigger coverage, the threshold, and two structural indicators, and
predicts a candidate's aggregate margin. It never receives the
protected-trigger success rate. The masked-core failures reported below are
therefore produced by a proposer optimising the quantity it can observe, not
by an adversary constructed to defeat the governor; responsibility for the
protected group rests with the certification rule, not with the proposer.

All experiments use one selected candidate and one certification attempt per
active version, as required by the protocol in Figure~\ref{fig:governor-protocol}. Censored
runs are those in which the candidate is not certified by the finite
horizon. Wilson intervals below quantify Monte Carlo variation across
replications; they are not confidence-sequence bounds. The accompanying
script fixes all seeds and regenerates the reported tables, figures, and
per-run outputs.

\subsection{RQ1: optional continuation and finite-sample latency}
\label{sec:rq1-experiment}

\paragraph{Question.}
Does fresh post-selection sampling prevent prefix selection from entering
the certificate, and what finite-sample latency is introduced by
anytime-valid inspection?

This experiment illustrates optional continuation and the
conservativeness of the Hoeffding bound; it is not a separate contribution
and is reported as a single table.

An independent prefix selects the empirically best of eight candidates;
the compared rules then receive identical fresh Bernoulli outcomes. Across
1,000 replications, we compare direct activation, repeated misuse of a
fixed-time exact test, an exact test with spending, and the paper's Hoeffding
rule. The complete protocol and seeds are provided with the artifact.

Table~\ref{tab:rq1-certification} summarizes the activation frequencies and
the median numbers of completed post-selection obligations at activation.
\begin{table}[t]
\centering
\small
\caption{RQ1 synthetic post-selection certification. Each row has 1000 independent replications and 6,000 completed-obligation opportunities. Wilson intervals measure Monte Carlo variation; $q$ is the median completed count at activation.}
\label{tab:rq1-certification}
\resizebox{\linewidth}{!}{%
\begin{tabular}{@{}llrrrr@{}}
\toprule
True $p$ & Procedure & Activations & Wilson 95\% (\%) & Median $q$ & Censored \\
\midrule
0.890 & Direct proposal & 1000/1000 & [99.6, 100.0] & 0 & 0 \\
 & Fixed-time, repeatedly tested & 56/1000 & [4.3, 7.2] & 107 & 944 \\
 & Exact test with spending & 0/1000 & [0.0, 0.4] & -- & 1000 \\
 & Paper's Hoeffding bound & 0/1000 & [0.0, 0.4] & -- & 1000 \\
0.899 & Direct proposal & 1000/1000 & [99.6, 100.0] & 0 & 0 \\
 & Fixed-time, repeatedly tested & 190/1000 & [16.7, 21.5] & 249 & 810 \\
 & Exact test with spending & 0/1000 & [0.0, 0.4] & -- & 1000 \\
 & Paper's Hoeffding bound & 0/1000 & [0.0, 0.4] & -- & 1000 \\
0.980 & Direct proposal & 1000/1000 & [99.6, 100.0] & 0 & 0 \\
 & Fixed-time, repeatedly tested & 1000/1000 & [99.6, 100.0] & 54 & 0 \\
 & Exact test with spending & 1000/1000 & [99.6, 100.0] & 240 & 0 \\
 & Paper's Hoeffding bound & 1000/1000 & [99.6, 100.0] & 1,461 & 0 \\
\bottomrule
\end{tabular}}
\end{table}

At the below-threshold probabilities $p=0.890$ and $p=0.899$, direct proposal
accepts every candidate, while repeated use of the unadjusted fixed-level test
activates in 56/1,000 and 190/1,000 runs. Neither spent procedure activates
in these simulations. At the above-threshold probability $p=0.98$, all three
statistical procedures eventually activate, but their median completed
sample sizes are $54$, $240$, and $1{,}461$ for the repeatedly inspected
fixed-level test, the spent exact test, and the Hoeffding bound,
respectively. The experiment illustrates the protection required for
sequential inspection and the finite-sample cost of the Hoeffding bound;
it is not a separate theoretical contribution.

\subsection{RQ2: masked-core failure}
\label{sec:rq2-experiment}

\paragraph{Question.}
Can an aggregate-oriented AI proposal pass aggregate certification while
systematically failing on protected triggers, and does joint certification block it?

Write $p^{\mathrm{core}}$ for the protected-group success probability and
suppose non-protected responses always succeed. Equation~\eqref{eq:core-mixture}
then gives
\[
p=1-w+w p^{\mathrm{core}}.
\]
Asymptotic aggregate certification requires $p>\lambda$,
whereas core-conditional support fails when
$p^{\mathrm{core}}<\lambda_{\mathrm{core}}$. The masked-core region is
therefore
\begin{equation}
\label{eq:masked-window}
\left\{p^{\mathrm{core}}\in[0,1]:
p^{\mathrm{core}}>1-\frac{1-\lambda}{w},\quad
p^{\mathrm{core}}<\lambda_{\mathrm{core}}\right\}.
\end{equation}
Its width is
\[
\max\!\left\{0,\lambda_{\mathrm{core}}-
\max\!\left(0,1-\frac{1-\lambda}{w}\right)\right\}.
\]
The width increases as the lower endpoint decreases but saturates when that
endpoint reaches zero; unlike the sampling latency of
Remark~\ref{rem:core-certification-cost}, it does not globally scale as
$1/w$.

\paragraph{Varying the protected share.}
At event times $t_n=n$, trigger $A$ occurs every three events. Conditional
on $A$, protected predicate $C$ occurs independently with probability $w$.
We test $w=0.01$, $0.03$, $0.05$, and $0.07$. Candidate response $B$ is checked at the
next event, with
$\lambda=\lambda_{\mathrm{core}}=0.9$. Every protected response fails and
every non-protected response succeeds, so the aggregate probability is
$1-w\ge0.93$. A drift declaration is supplied exogenously to isolate the
certification mechanism. A frozen random-forest proposer, trained on 4,000
independent synthetic tasks, selects between the admissible full trigger and
its core-only restriction from a separate prefix. Its features contain the
aggregate empirical frequency, trigger coverage, threshold, and structural
indicators, but no protected-group outcome frequency. Aggregate-only and joint certification see the same 200
streams at each $w$, up to 45,000 event indices. The proposer selects the
full-trigger candidate in all 200 independent prefixes at every tested value
of $w$.

This constructed masked-core setting makes the full trigger attractive in
aggregate, so it is not by itself an informative accuracy test for the
proposer. We therefore audit the same frozen model on 500 independent,
deliberately closer four-candidate ranking tasks, varying both adaptive
trigger retention and an adaptive response conjunct. Table~\ref{tab:rq2-proposer-audit}
reports 65 ranking errors and mean true-margin regret $0.0035$. These errors
make the proposer genuinely fallible while remaining separate from the
certification samples; proposer accuracy is descriptive and is not an
assumption of the guarantee.

\begin{table}[t]
\centering
\small
\caption{Independent ranking audit of the frozen RQ2 proposer.}
\label{tab:rq2-proposer-audit}
\begin{tabular}{@{}rrrr@{}}
\toprule
Tasks & Exact rankings & Ranking errors & Mean margin regret \\
\midrule
500 & 435 & 65 & 0.0035 \\
\bottomrule
\end{tabular}
\end{table}

\begin{figure}[t]
\centering
\includegraphics[width=\linewidth]{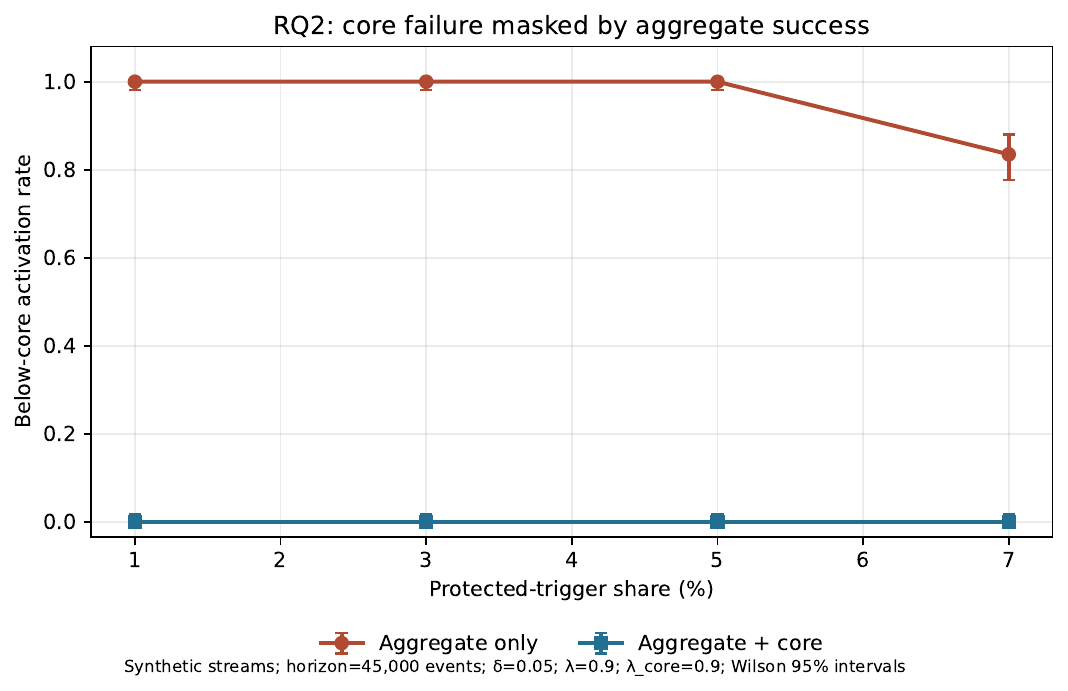}
\caption{Activations within the finite horizon when protected responses
always fail. Each point uses 200 paired streams; bars show Wilson 95\%
intervals.}
\Description{Core-unsupported activation rate against protected-trigger share. The aggregate-only curve is high, whereas the aggregate-plus-core curve remains at zero.}
\label{fig:rq2-masked-core}
\end{figure}
\begin{table}[t]
\centering
\small
\caption{RQ2 masked-core experiment with proposals selected by the frozen supervised proposer. A below-core activation occurs while $p^{\mathrm{core}}<\lambda_{\mathrm{core}}$; brackets give Wilson 95\% intervals.}
\label{tab:rq2-masked-core}
\begin{tabular}{@{}lrrrr@{}}
\toprule
$w$ & Certification rule & Below core & Wilson 95\% & Median event \\
\midrule
0.01 & Aggregate only & 200/200 & [98.1, 100.0] & 3,305 \\
0.01 & Aggregate + core & 0/200 & [0.0, 1.9] & -- \\
0.03 & Aggregate only & 200/200 & [98.1, 100.0] & 5,795 \\
0.03 & Aggregate + core & 0/200 & [0.0, 1.9] & -- \\
0.05 & Aggregate only & 200/200 & [98.1, 100.0] & 12,329 \\
0.05 & Aggregate + core & 0/200 & [0.0, 1.9] & -- \\
0.07 & Aggregate only & 167/200 & [77.7, 88.0] & 36,158 \\
0.07 & Aggregate + core & 0/200 & [0.0, 1.9] & -- \\
\bottomrule
\end{tabular}
\end{table}

Figure~\ref{fig:rq2-masked-core} and
Table~\ref{tab:rq2-masked-core} show aggregate-only activation in every run
for $w\le0.05$ and in 167/200 runs for $w=0.07$. joint certification never
activates because its core lower bound remains zero. The 33 aggregate-only
runs censored at $w=0.07$ reflect the smaller finite-horizon margin
$1-w-\lambda=0.03$, not absence of the masked-core phenomenon.

\paragraph{Varying protected success.}
We next fix $w=0.2$ and test $p^{\mathrm{core}}$ at $0$, $0.5$, $0.8$,
$0.88$, $0.92$, $0.95$, and $0.99$, using 100 paired replications and a
horizon of 150,000 event indices. Equation~\eqref{eq:masked-window} gives
the open theoretical interval $(0.5,0.9)$.

\begin{figure}[t]
\centering
\includegraphics[width=\linewidth]{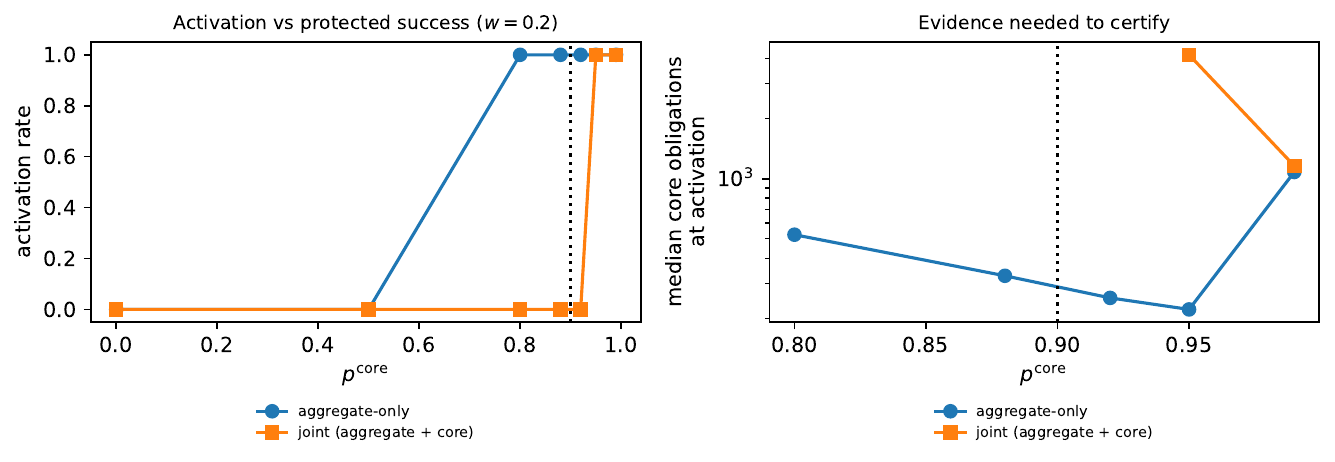}
\caption{Activation rates against protected success at $w=0.2$ (left) and
median completed protected obligations at joint activation (right). The
dotted line is $\lambda_{\mathrm{core}}$.}
\Description{Two plots show activation rates and median core evidence as the protected success probability increases.}
\label{fig:rq2-gradient}
\end{figure}
\begin{table}[t]
\centering
\small
\caption{RQ2 protected-success sweep at $w=0.20$. A dash means that no run activated within the horizon.}
\label{tab:rq2-gradient}
\begin{tabular}{@{}lrrrr@{}}
\toprule
$p^{\mathrm{core}}$ & Aggregate-only activations & Joint activations & Joint median core $q$ & Joint censored \\
\midrule
0.00 & 0/100 & 0/100 & -- & 100 \\
0.50 & 0/100 & 0/100 & -- & 100 \\
0.80 & 100/100 & 0/100 & -- & 100 \\
0.88 & 100/100 & 0/100 & -- & 100 \\
0.92 & 100/100 & 0/100 & -- & 100 \\
0.95 & 100/100 & 100/100 & 4,177 & 0 \\
0.99 & 100/100 & 100/100 & 1,160 & 0 \\
\bottomrule
\end{tabular}
\end{table}

Figure~\ref{fig:rq2-gradient} and
Table~\ref{tab:rq2-gradient} report the activation frequencies and the
completed core sample sizes across the protected-success sweep.
At $0.8$ and $0.88$, inside the window, aggregate-only certification
activates in 100/100 runs and joint certification in none. The value $0.5$ is the
aggregate boundary and neither rule activates. The value $0.92$ is outside
the masked-core region and is genuinely core-compliant, but remains censored
under joint certification at this finite horizon, whose Hoeffding bound is conservative near the threshold. At $0.95$ and
$0.99$, joint certification activates in every run, after medians of 4,177 and
1,160 protected obligations respectively. Thus censoring near the threshold
is a finite-horizon liveness effect, not evidence that the candidate
violates the protected requirement.

\subsection{RQ3: sequential governed evolution}
\label{sec:rq3-experiment}

\paragraph{Question.}
Does joint certification continue to prevent masked-core activations when an
AI proposer repeatedly revises the specification over one evolving history?

The third experiment exercises the complete loop. The two-window detector
of Section~\ref{sec:drift} initiates proposal; the frozen supervised proposer
ranks the formally admissible neighbourhood; and the governor applies
$\delta_j=\delta/2^{j+1}$ across successive transitions. Candidate generation
changes at most one adaptive structural component per transition. Open
obligations retain the trigger, response, and temporal parameters of their
origin version.

Across four regimes the protected success probability follows
$0.99\to0.97\to0.80\to0.55$, while non-protected responses remain easy.
We use $w=0.3$, horizon 120,000,
$\lambda_{\mathrm{core}}=0.9$, parameter envelope
$\{0.95,0.92,0.90\}$, and 100 paired histories. The proposer predicts each
candidate's aggregate margin from prefix-level aggregate features and never
receives the protected success frequency. Training data, model seed, and
simulation seeds are fixed independently of the certification samples.

\begin{figure}[t]
\centering
\includegraphics[width=\linewidth]{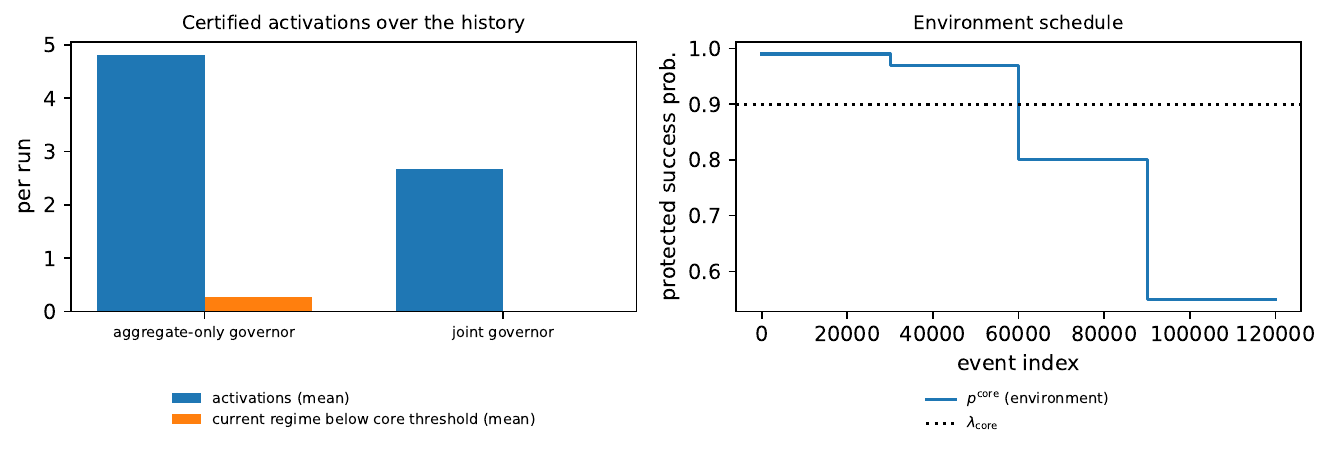}
\caption{Sequential experiment. Left: mean certified activations and
activations occurring while the synthetic current regime is below the core
threshold. Right: protected-success schedule.}
\Description{Aggregate-only certification activates in below-core regimes; joint certification does not. The protected success schedule decreases across four regimes.}
\label{fig:rq3-endtoend}
\end{figure}
\begin{table}[t]
\centering
\small
\caption{RQ3 sequential experiment with the frozen supervised AI proposer. `Below core' counts runs with at least one activation during a regime where $p^{\mathrm{core}}<\lambda_{\mathrm{core}}$.}
\label{tab:rq3-sequential}
\resizebox{\linewidth}{!}{%
\begin{tabular}{@{}lrrrr@{}}
\toprule
Certification rule & Mean activations & Mean proposals & Below core & Sample-target failure \\
\midrule
Aggregate only & 4.82 & 5.82 & 27/100 & 0/100 \\
Aggregate + core & 2.68 & 3.68 & 0/100 & 0/100 \\
\bottomrule
\end{tabular}}
\end{table}

Figure~\ref{fig:rq3-endtoend} and
Table~\ref{tab:rq3-sequential} summarize the activation counts and the
protected-success schedule across the paired histories.
Aggregate-only certification produces 4.82 activations per history on
average and activates while the current regime is below the protected
threshold in 27/100 histories. joint certification produces 2.68 activations
on average and never does so in these runs. This is the masked-core mechanism
operating across a sequence rather than at one decision.

The two diagnostics in Table~\ref{tab:rq3-sequential} must not be conflated.
The lifetime theorem controls activation when the average predictable mean
of the certification sample is below its threshold; no such sample-target
failure was observed under either certification rule. Whether the synthetic current
regime is below the core threshold is a different operational diagnostic,
because a certification sample may cross a regime boundary. The experiment
therefore illustrates repeated governed decisions under the theorem's
schedule; it neither estimates the lifetime error probability nor proves
future operational validity.

\subsection{Monitored clinical-alarm case study and runtime cost}
\label{sec:clinical-case-study}

This case study instantiates the clinical-alarm scenario of
Section~\ref{sec:introduction} with an executable bounded-response monitor,
rather than sampling obligation outcomes directly. Proposition $A$ denotes a
desaturation alarm, $C$ the protocol-critical subset, and $B$ a documented
intervention. Alarm origins are separated by more than the largest response
horizon so that each intervention is attributable to one obligation. The
trace contains 30,000 alarms; protected alarms occur with probability $0.1$,
and full-response probabilities are $0.60$ and $0.99$ for protected and
non-protected alarms, respectively.

The executable monitor evaluates the incumbent and candidate formulas
\[
\psi_0=A\Rightarrow_{0.9,[1,3]}B,
\qquad
\psi^\star=A\Rightarrow_{0.9,[1,8]}B,
\]
with designer-protected trigger $A\land C$ and protected response $B$.

Starting from the incumbent window $[1,3]$, the frozen supervised proposer
ranks an admissible $[1,8]$ revision using an independent 1,000-obligation
prefix. A second proposed revision that removes the protected response $B$
is rejected by structural admissibility before statistical ranking. The
selected candidate is then monitored prospectively on the Boolean timed
trace. Table~\ref{tab:clinical-case} reports the resulting certificates and
the median of nine reference runs of the executable monitor.

\begin{table}[t]
\centering
\small
\caption{Monitored clinical-alarm case study and reference runtime.}
\label{tab:clinical-case}
\begin{tabularx}{\linewidth}{@{}Xr@{}}
\toprule
Quantity & Result \\
\midrule
Completed alarm obligations & 30,000 \\
Observed protected share & 0.099 \\
Candidate aggregate success & 0.953 \\
Candidate protected success & 0.608 \\
Aggregate-only activation sample & 3,321 \\
Joint activation & not activated \\
Pending outcome: origin / retroactive & 1 / 0 \\
Baseline / governed median runtime & 72.0 / 217.6 ms \\
Baseline / governed time per event & 0.240 / 0.725 $\mu$s \\
Baseline / governed peak memory & 30.9 / 60.3 KiB \\
Baseline / governed peak pending & 1 / 2 \\
Parallel-governor increment & 145.6 ms \\
Reference runtime ratio & 3.02$\times$ \\
\bottomrule
\end{tabularx}
\end{table}

The widened candidate has aggregate success $0.953$ but protected success
$0.608$. Aggregate-only certification activates after 3,321 completed
post-selection obligations; joint certification never activates because the
core conjunct remains unsupported. A pending alarm whose intervention occurs
six events after its origin succeeds under its origin window $[1,8]$ but
would fail under a retroactively substituted $[1,3]$ window, producing the
reported $1/0$ outcome pair and exercising
Theorem~\ref{thm:origin-nonretroactivity}. The reference timings in the table
compare the non-adaptive incumbent monitor with concurrent incumbent and
candidate monitors plus group classification, sequential lower bounds, and
the activation check. The governed run uses $0.725$ rather than $0.240$
microseconds per trace event, peaks at $60.3$ rather than $30.9$ KiB of
Python-tracked allocations, and has at most two rather than one pending
obligations because the two monitors run in parallel. Absolute timings are
machine-dependent; the artifact regenerates all measurements and retains the
obligation counts and decisions as reproducible outputs.

The case study therefore exercises the actual obligation interface and both
main failure mechanisms. It remains semi-synthetic: its probabilities encode
a realistic alarm workflow but are not estimates from patient data.

\subsection{Artifact and reproducibility statement}
\label{sec:artifact-statement}

The reproducibility artifact is publicly available at
\url{https://github.com/ruggerolanotte/protected-cores-artifact}.
The artifact contains one Python program and requires Python 3.10+, NumPy,
SciPy, Matplotlib, and scikit-learn. The command consists of
\texttt{python scripts/fac\_experiments.py}, followed by one of
\texttt{rq1}, \texttt{rq2}, \texttt{rq3}, or \texttt{case}, and
\texttt{--out DIR}; these targets
regenerate every manuscript table and figure; the \texttt{all} target runs
the complete suite and \texttt{all --quick} performs an interface smoke test.
All random seeds, model-training sizes, horizons, and confidence allocations
are command-line parameters with manuscript defaults. RQ2 and RQ3 archive the
paired per-run records, while the case study archives its monitored counts,
activation indices, pending-obligation regression result, and reference
timings.

\section{Conclusion}

We presented a framework for the certified evolution of temporal
specifications under an untrusted AI proposer. The proposer suggests
parametric or structural revisions; a symbolic governor checks their
admissibility and requires post-selection evidence before activation.

Two results organize the framework. Origin-version semantics prevents a
revision from changing the outcome of obligations generated under an
earlier version. The masked-core counterexample shows that aggregate
certification can admit a candidate that fails systematically on protected
triggers; adding core-conditional certification addresses this failure
mode. A protected-envelope invariant supports admissible revisions, while
a proposer-independent lifetime theorem bounds erroneous certifications
relative to the predictable means of completed samples, including when
aggregate and protected samples overlap.

Controlled experiments and a monitored clinical-alarm case study use a
frozen supervised proposer to illustrate the masked-core failure at one
decision and across repeated governed transitions, while also measuring the
cost of prospective certification. Section~\ref{sec:activation-scope} states the
condition under which a certificate transfers to obligations generated after
activation, and records that it must be assumed rather than inferred: two
laws may agree on the entire observed prefix and differ afterwards. That
condition is not a third activation guarantee, and the experiments do not
infer it from a finite prefix.

\begin{acks}
The authors acknowledge the use of ChatGPT (OpenAI) during the preparation of this manuscript to assist with manuscript organization, preliminary reformulation of selected definitions and technical arguments, and language revision. All technical content, proofs, and references were independently reviewed and verified by the authors, who take full responsibility for the final manuscript.
\end{acks}

\bibliographystyle{ACM-Reference-Format}
\bibliography{references}

\appendix
\section{Proofs}
\label{app:proofs}

\subsection{Proof of Theorem~\ref{thm:core-envelope-preserving}}
\label{app:proof-thm-core-envelope-preserving}
\begin{proof}
The claim is immediate for $i=0$. Suppose it holds for $\psi_i$ and
$\psi_i\mathrel{\mathcal R}\psi_{i+1}$. By
Definition~\ref{def:admissible-revision}, the two protected components are
copied unchanged into $\psi_{i+1}$ and its parameter tuple remains in
$\mathcal A_{\mathrm{par}}$. Moreover,
$\operatorname{trig}(\psi_{i+1})
=
\varphi_1^{\mathrm{core}}
\lor
\varphi_{1,i+1}^{\mathrm{adapt}}$,
so every occurrence of the protected trigger satisfies the derived trigger,
and
$\operatorname{resp}(\psi_{i+1})
=
\varphi_2^{\mathrm{core}}
\land
\varphi_{2,i+1}^{\mathrm{adapt}}$,
so every successful derived response satisfies the protected response.
Induction completes the proof.
\end{proof}

\subsection{Proof of Theorem~\ref{thm:origin-nonretroactivity}}
\label{app:proof-thm-origin-nonretroactivity}
\begin{proof}
The trigger is evaluated as
$\mathbf e,k\models\operatorname{trig}(\psi_j)$.
When it holds, the outcome is
\[
X_k^{\psi_j}
=
\mathbf 1\!\left[
\exists m\ge k:
t_m-t_k\in[a_j,b_j]
\land
\mathbf e,m\models\operatorname{resp}(\psi_j)
\right].
\]
The completion index is $c_{\psi_j}(k)$ as defined in the statement; it is
finite because $H(\psi_j)<\infty$ and timed streams are non-Zeno.
All three expressions use the fixed origin specification
$\psi_j$ and the same timed stream $\mathbf e$; none depends
on a subsequent revision. 
Each origin index has a unique active
specification by the activation-interval convention, which establishes
unique attribution as well.
\end{proof}

\subsection{Proof of Proposition~\ref{prop:masked-core-failure}}
\label{app:proof-prop-masked-core-failure}

\begin{proof}
Let $A,C,B\in\AP$, and consider the protected decomposition
\[
\varphi_1^{\mathrm{core}}
=
A\land C,
\qquad
\varphi_1^{\mathrm{adapt}}
=
A\land\neg C,
\qquad
\varphi_2^{\mathrm{core}}
=
B,
\qquad
\varphi_2^{\mathrm{adapt}}
=
\top.
\]
Take
\[
\psi^\star
:=
\bigl((A\land C)\lor(A\land\neg C)\bigr)
\Rightarrow_{0.9,[a,b]}
(B\land\top).
\]
Then
\[
\operatorname{trig}(\psi^\star)
=
(A\land C)\lor(A\land\neg C)
=
A,
\qquad
\operatorname{resp}(\psi^\star)
=
B.
\]
Thus each obligation is generated exactly when $A$ holds. We may therefore
restrict the remainder of the argument to the sequence of $A$-triggered
obligations. Let the designer-fixed protected threshold be
$\lambda_{\mathrm{core}}=0.9$.

We now construct an explicit timed-stream law. Place the origins of
$A$-triggered obligations at deterministic times separated by more than
$H(\psi^\star)$. This prevents the response assigned to one obligation from
affecting any other obligation. Choose regime $R_0$ to contain a nonempty
finite collection of such obligations, each completed within $R_0$, and
let $R_1$ contain an infinite sequence of completed obligations with origins
$k_1<k_2<\cdots$.

For every $A$-triggered obligation in either regime, independently draw
$G_s\sim\operatorname{Bern}(0.01)$
and set $C$ true at its origin if and only if $G_s=1$. Hence
\[
G_s
=
\begin{cases}
1, & \text{for a protected trigger }A\land C,\\
0, & \text{for an adaptive-only trigger }A\land\neg C,
\end{cases}
\]
and protected triggers have conditional share
\[
\Pr(G_s=1)
=
\Pr\!\left(
\mathbf e,k_s\models C
\mid
\mathbf e,k_s\models A
\right)
=
0.01.
\]

In regime $R_0$, independently of the group indicators, draw
$Y_s\sim\operatorname{Bern}(0.99)$
and make $B$ true within the response window of obligation $s$ if and only
if $Y_s=1$. Consequently,
\[
X_{k_s}^{\psi^\star}=Y_s
\]
in $R_0$, independently of whether the trigger is protected or
adaptive-only. Therefore,
\[
\Pr_{R_0}\!\left(
X_{k_s}^{\psi^\star}=1\mid G_s=1
\right)
=
\Pr_{R_0}\!\left(
X_{k_s}^{\psi^\star}=1\mid G_s=0
\right)
=
0.99.
\]
It follows that
\[
p_{R_0}^{\mathrm{core}}(\psi^\star)
=
0.99
\qquad\text{and}\qquad
p_{R_0}(\psi^\star)
=
0.01\cdot0.99
+
0.99\cdot0.99
=
0.99.
\]

In regime $R_1$, use the same group law but make $B$ true within the
response window of obligation $s$ if and only if $G_s=0$. Equivalently,
\[
X_{k_s}^{\psi^\star}
=
1-G_s.
\]
Thus every protected-trigger obligation fails and every adaptive-only
obligation succeeds:
\[
\Pr_{R_1}\!\left(
X_{k_s}^{\psi^\star}=1\mid G_s=1
\right)
=
0,
\qquad
\Pr_{R_1}\!\left(
X_{k_s}^{\psi^\star}=1\mid G_s=0
\right)
=
1.
\]
Hence
\[
p_{R_1}^{\mathrm{core}}(\psi^\star)
=
0,
\]
whereas the aggregate probability remains
\[
p_{R_1}(\psi^\star)
=
0.01\cdot0
+
0.99\cdot1
=
0.99
=
p_{R_0}(\psi^\star).
\]

Suppose now that candidate selection occurs before any outcome generated in
$R_1$ is observed. Since the variables $(G_s)_{s\ge1}$ are independent
Bernoulli variables with parameter $0.01$, the aggregate outcomes
\[
\bigl(X_{k_s}^{\psi^\star}\bigr)_{s\ge1}
=
(1-G_s)_{s\ge1}
\]
are independent Bernoulli variables with parameter $0.99$. By the strong
law of large numbers,
\[
\widehat p_q
=
\frac{1}{q}
\sum_{s=1}^{q}
X_{k_s}^{\psi^\star}
\longrightarrow
0.99
\qquad\text{almost surely}.
\]
For every fixed aggregate error allocation
$\delta_j^{\mathrm{all}}>0$, the radius of the aggregate lower bound
satisfies
\[
\sqrt{
\frac{
\log\!\bigl(
\pi^2q^2/(6\delta_j^{\mathrm{all}})
\bigr)
}{
2q
}}
\longrightarrow
0.
\]
Consequently,
\[
L_{\tau_q}>0.9
\]
for all sufficiently large $q$, almost surely. Aggregate-only
certification therefore eventually certifies $\psi^\star$ almost surely.

Because
$\Pr(G_s=1)=0.01>0$, protected triggers occur infinitely often almost
surely. Every outcome in the resulting core subsequence equals zero.
Therefore, for every nonempty core sample,
\[
\widehat p_q^{\mathrm{core}}
=
0,
\qquad
L_{\tau_q^{\mathrm{core}}}^{\mathrm{core}}
=
0
<
\lambda_{\mathrm{core}}.
\]
Before the first protected-trigger obligation is completed,
$N_m^{\mathrm{core}}=0$, so the nonempty-core-sample requirement of
Definition~\ref{def:statistical-certification} fails. After the first such
completion, the core lower bound remains zero at every event index.
Hence the joint aggregate and core-conditional criterion never certifies
$\psi^\star$.
\end{proof}

\subsection{Proof of Proposition~\ref{prop:anytime-lower-bound}}
\label{app:proof-prop-anytime-lower-bound}
\begin{proof}
Fix an enabled $g$ and condition on $\mathcal H_j$, which fixes the
candidate, protected-trigger predicate, and budget. All selected obligations
share the finite horizon $H^\star$; selected origin indices are increasing,
and timestamps are non-decreasing. Completion by event index $m$ of a
selected obligation therefore implies completion of every earlier selected
obligation in the same group. Consequently the outcomes available at $m$ are
exactly
$X_{k_1^{(g)}},\ldots,X_{k_{N_m^{(g)}}^{(g)}}$.

By definition of the completed-outcome filtration and conditional mean,
\[
\mathbb E\!\left[
X_{k_s^{(g)}}-p_s^{(g)}
\mid\mathcal G_{j,s-1}^{(g)}
\right]=0,
\qquad
X_{k_s^{(g)}}-p_s^{(g)}
\in[-p_s^{(g)},1-p_s^{(g)}].
\]
The predictable interval has width one. Conditional Hoeffding for martingale
differences therefore gives, for every fixed $q\ge1$,
\[
\Pr\!\left(
\widehat p_q^{(g)}-\bar p_q^{(g)}
>
\sqrt{\frac{\log(1/\delta_{j,q}^{(g)})}{2q}}
\;\middle|\;\mathcal H_j
\right)
\le\delta_{j,q}^{(g)}.
\]
A union bound over $q\ge1$ costs
$\sum_{q\ge1}\delta_{j,q}^{(g)}=\delta_j^{(g)}$. On the complementary
event the inequality holds for every $q$ simultaneously.

For an enabled group $g$, let
\[
\mathcal I_g
:=
\bigcap_{q\ge1}
\left\{
\bar p_q^{(g)}
\ge
\widehat p_q^{(g)}
-
\sqrt{
\frac{\log\!\bigl(1/\delta_{j,q}^{(g)}\bigr)}
     {2q}
}
\right\}.
\]
The preceding argument gives
$\Pr(\mathcal I_g\mid\mathcal H_j)
\ge
1-\delta_j^{(g)}$.

Fix a sample path in $\mathcal I_g$ and an event index $m>n^\star$ with
$N_m^{(g)}>0$. Set $q=N_m^{(g)}$ on that sample path. By
Equation~\eqref{eq:initial-segment}, the outcomes available at $m$ are
exactly the first $q$ outcomes of group $g$. Hence
\[
\bar p_{N_m^{(g)}}^{(g)}
\ge
\widehat p_{N_m^{(g)}}^{(g)}
-
\sqrt{
\frac{
\log\!\bigl(1/\delta_{j,N_m^{(g)}}^{(g)}\bigr)
}{
2N_m^{(g)}
}
}.
\]
Moreover, $\bar p_{N_m^{(g)}}^{(g)}\ge0$ because it is an average of
conditional probabilities. Therefore,
\[
\bar p_{N_m^{(g)}}^{(g)}
\ge
\max\left\{
0,\,
\widehat p_{N_m^{(g)}}^{(g)}
-
\sqrt{
\frac{
\log\!\bigl(1/\delta_{j,N_m^{(g)}}^{(g)}\bigr)
}{
2N_m^{(g)}
}
}
\right\}
=
L_m^{(g)}.
\]
Since the sample path and $m$ were arbitrary, the inequality holds
simultaneously for every $m>n^\star$ with $N_m^{(g)}>0$. This pathwise
substitution requires no predictability of future origin indices.
\end{proof}

\subsection{Proof of Proposition~\ref{prop:joint-core-coverage}}
\label{app:proof-prop-joint-core-coverage}
\begin{proof}
For each enabled $g$, Proposition~\ref{prop:anytime-lower-bound} bounds
its conditional failure probability by $\delta_j^{(g)}$.
A union bound over the at most two enabled groups gives
$\delta_j^{\mathrm{all}}+\delta_j^{\mathrm{core}}\le\delta_j$.
The two groups may overlap; independence is not used.
\end{proof}

\subsection{Proof of
Corollary~\ref{cor:joint-regime-level-certification}}
\label{app:proof-cor-joint-regime-level-certification}

\begin{proof}
Let $\mathcal C_j$ denote the simultaneous coverage event of
Proposition~\ref{prop:joint-core-coverage}. Conditionally on
$\mathcal H_j$,
$\Pr(\mathcal C_j\mid\mathcal H_j)
\ge
1-\delta_j$.

Fix a sample path in $\mathcal C_j$ on which $n^\star<m$ and
$\psi^\star$ is statistically certified at $m$. By
Definition~\ref{def:statistical-certification}, for every enabled group
$g$,
$N_m^{(g)}>0$ and
$L_m^{(g)}\ge\lambda_g$.
Because the sample path belongs to $\mathcal C_j$,
\[
\bar p_{N_m^{(g)}}^{(g)}
\ge
L_m^{(g)}
\ge
\lambda_g
\]
for every enabled group $g$.

For every
$s=1,\ldots,N_m^{(g)}$, the event
$\{n^\star<m,\;N_m^{(g)}\ge s\}$ holds on this sample path. The alignment
assumption therefore gives
$p_s^{(g)}
=
p_{R_r}^{(g)}(\psi^\star)$
almost surely. Consequently,
\[
\bar p_{N_m^{(g)}}^{(g)}
=
\frac{1}{N_m^{(g)}}
\sum_{s=1}^{N_m^{(g)}}p_s^{(g)}
=
p_{R_r}^{(g)}(\psi^\star).
\]
Hence, on $\mathcal C_j$, statistical certification at $m$ implies
$p_{R_r}^{(g)}(\psi^\star)
\ge
\lambda_g$
for every enabled group $g$.

It follows that the event
\[
\left\{
\begin{array}{c}
n^\star<m,\;
\psi^\star\text{ is statistically certified at }m,\\
\exists\,g\in\{\mathrm{all},\mathrm{core}\}\text{ enabled}:
p_{R_r}^{(g)}(\psi^\star)<\lambda_g
\end{array}
\right\}
\]
is contained, up to a null set, in $\mathcal C_j^c$. Therefore,
conditionally on $\mathcal H_j$, its probability is at most
$\delta_j$.
\end{proof}

\subsection{Proof of Theorem~\ref{thm:lifetime-certification}}
\label{app:proof-thm-lifetime-certification}
\begin{proof}
Fix $j$. If the version is never reached or $n^\star=\infty$,
there is no activation and $F_j=\varnothing$. On paths with finite selection time, use the simultaneous aggregate
and (when enabled) core coverage event. Define it globally by
\[
\mathcal C_j
=\{n^\star=\infty\}
\cup\left\{n^\star<\infty\;\middle|\;
\begin{array}{l}
\forall m>n^\star\text{ with }N_m>0:\\
\qquad
\bar p_{N_m}
\ge L_m;\\[1mm]
\text{and, if core certification is enabled,}\\
\forall m>n^\star\text{ with }N_m^{\mathrm{core}}>0:\\
\qquad
\bar p^{\mathrm{core}}_{N_m^{\mathrm{core}}}
\ge L_m^{\mathrm{core}}
\end{array}
\right\}.
\]
The displayed inequalities are required only for their respective
nonempty samples. Thus $\mathcal C_j$ holds automatically when no
candidate is selected, while its core condition is omitted only when
core certification is disabled. By
Proposition~\ref{prop:joint-core-coverage},
\[
\mathbf 1_{\{n^\star<\infty\}}
\Pr(\mathcal C_j^c\mid\mathcal H_j)
\le
\mathbf 1_{\{n^\star<\infty\}}\delta_j
\qquad\text{a.s.}
\]

If activation occurs, its first-crossing definition gives
$m_j^{\mathrm{cert}}=n_{j+1}-1$, nonempty aggregate and, when enabled,
core samples, and both lower bounds at least their respective
thresholds. Since $\mathcal C_j$ holds at \emph{every} eligible
sample size, it holds in particular at this random decision index.
Consequently, on $\mathcal C_j$ neither
$F_j^{\mathrm{all}}$ nor $F_j^{\mathrm{core}}$ can occur, so
\[
F_j\subseteq
\{n^\star<\infty\}\cap\mathcal C_j^c.
\]
The selection indicator and the budget $\delta_j$ are
$\mathcal H_j$-measurable. Conditional expectation therefore gives
\[
\Pr(F_j)
\le \mathbb E\!\left[
\mathbf 1_{\{n^\star<\infty\}}
\Pr(\mathcal C_j^c\mid\mathcal H_j)
\right]
\le\mathbb E\!\left[
\mathbf 1_{\{n^\star<\infty\}}\delta_j
\right]
\le\mathbb E[\delta_j].
\]
Finally, countable subadditivity, Tonelli's theorem, and the
almost-sure lifetime budget give
\[
\Pr\!\left(\bigcup_{j=0}^{\infty}F_j\right)
\le\sum_{j=0}^{\infty}\Pr(F_j)
\le\sum_{j=0}^{\infty}\mathbb E[\delta_j]
=\mathbb E\!\left[\sum_{j=0}^{\infty}\delta_j\right]
\le\delta.
\]
No independence between attempts, versions, or the two
certification samples is used.
\end{proof}

\end{document}